\documentclass[11pt,letterpaper]{article}
\pdfoutput=1
\usepackage[utf8]{inputenc}
\usepackage[T1]{fontenc}
\usepackage[margin=1in]{geometry}
\usepackage{amsmath,amsthm,amssymb,thmtools,thm-restate,booktabs,color,doi,graphicx,latexsym,url,xcolor,xspace}
\usepackage{bbold} %
\usepackage[numbers,sort&compress]{natbib}
\usepackage{microtype,hyperref}
\usepackage[nameinlink,capitalise]{cleveref}
\usepackage[inline]{enumitem}
\usepackage{algorithm}
\usepackage{algpseudocode}
\setlist[itemize]{label=--}
\setlist[enumerate]{label=(\arabic*),labelindent=\parindent,leftmargin=*}
\usepackage{array}
\usepackage{mathtools}
\usepackage{sectsty}
\allsectionsfont{\boldmath}

\definecolor{citecolor}{HTML}{0000C0}
\definecolor{urlcolor}{HTML}{000080}

\hypersetup{
    colorlinks=true,
    linkcolor=black,
    citecolor=citecolor,
    filecolor=black,
    urlcolor=urlcolor,
    pdftitle={Space-efficient population protocols for exact majority on general graphs}, %
    pdfauthor={Joel Rybicki, Jakob Solnerzik, Olivier Stietel, Robin Vacus} %
}

\newtheorem{theorem}{Theorem}
\newtheorem{lemma}[theorem]{Lemma}
\newtheorem{corollary}[theorem]{Corollary}
\newtheorem{definition}[theorem]{Definition}

\newenvironment{myabstract}
{\list{}{\listparindent 1.5em%
		\itemindent    \listparindent
		\leftmargin    1cm
		\rightmargin   1cm
		\parsep        0pt}%
	\item\relax}
{\endlist}

\newenvironment{mycover}
{\list{}{\listparindent 0pt
		\itemindent    \listparindent
		\leftmargin    1cm
		\rightmargin   1cm
		\parsep        0pt}%
	\raggedright
	\item\relax}
{\endlist}

\newcommand{\myemail}[1]{\,$\cdot$\, {\small #1}}
\newcommand{\myaff}[1]{\,$\cdot$\, {\small #1}\par\medskip}

\Crefname{equation}{Eq.}{Eqs.}

\renewcommand{\vec}[1]{\mathbf{#1}}
\DeclareMathOperator*{\argmax}{arg\,max}

\newcommand{\E}{\mathbb{E}}

\DeclareMathOperator*{\poly}{poly}
\DeclareMathOperator*{\polylog}{polylog}

\DeclareMathOperator*{\dist}{dist}
\DeclareMathOperator*{\Geom}{Geom}
\DeclareMathOperator*{\spec}{Spec}

\newcommand{\norm}[1]{\left\lVert#1\right\rVert}

\newcommand{\lout}{\ell_{\textrm{out}}}

\renewcommand{\P}[1]{\Pr\!\left[{#1}\right]}

\newcommand{\one}[1]{\mathbb{1}_{#1}}

\newcommand{\tmix}{\tau_{\mathsf{mix}}}

\newcommand{\trel}{\tau_{\mathsf{rel}}}

\newcommand{\tickfreq}{\tau_{\mathsf{tick}}}
\newcommand{\textinction}{T_{\mathsf{ext}}}
\newcommand{\textinctioncont}{T_{\mathsf{ext}}^{\mathsf{cont}}}
\newcommand{\tclear}{T_{\mathsf{clr}}}

\newcommand{\tbroadcast}{T_{\mathsf{br}}}

\newcommand{\atype}{\mathsf{A}}
\newcommand{\btype}{\mathsf{B}}
\newcommand{\ctype}{\mathsf{C}}

\newcommand{\smallatype}{\mathsf{a}}
\newcommand{\smallbtype}{\mathsf{b}}

\newcommand{\stype}{\mathsf{S}}
\newcommand{\wtype}{\mathsf{W}}

\newcommand{\Exp}{\mathbb{E}}

\newcommand{\pa}[1]{\left( #1 \right)}
\newcommand{\br}[1]{\left[ #1 \right]}

\newcommand{\calA}{\mathcal{A}}
\newcommand{\calB}{\mathcal{B}}
\newcommand{\calC}{\mathcal{C}}

\begin{document}

\begin{mycover}
	{\huge\bfseries\boldmath Space-efficient population protocols for exact majority on general graphs \par}

	\bigskip
	\bigskip

\textbf{Joel Rybicki}
\myemail{joel.rybicki@hu-berlin.de}
\myaff{Humboldt University of Berlin}

\textbf{Jakob Solnerzik}
\myemail{jakob.solnerzik@hu-berlin.de}
\myaff{Humboldt University of Berlin}

\textbf{Olivier Stietel}
\myemail{olivier.stietel@hu-berlin.de}
\myaff{Humboldt University of Berlin}

\textbf{Robin Vacus}
\myemail{robin.vacus@hu-berlin.de}
\myaff{Humboldt University of Berlin}

\bigskip
\end{mycover}

\medskip
\begin{myabstract}
  \noindent\textbf{Abstract.} We study exact majority consensus in the population protocol model.
In this model,
the system is described by a graph $G = (V,E)$ with $n$ nodes, and in each time step, a scheduler samples uniformly at random a pair of adjacent nodes to interact.
In the exact majority consensus task, each node is given a binary input, and the goal is to design a protocol that almost surely reaches a stable configuration, where all nodes output the majority input value.

We give improved upper and lower bounds for exact majority in general graphs.
First, we give asymptotically tight time lower bounds  for general (unbounded space) protocols.
Second, we obtain new upper bounds parameterized by the relaxation time $\trel$ of the random walk on $G$ induced by the scheduler and the degree imbalance $\Delta/\delta$ of $G$.
Specifically, we give
a protocol that stabilizes in $O\left( \tfrac{\Delta}{\delta} \tau_{\mathsf{rel}} \log^2 n \right)$ steps in expectation and with high probability and uses  $O\left( \log n \cdot \left(  \log\left(\tfrac{\Delta}{\delta}\right) + \log \left(\tfrac{\trel}{n}\right) \right) \right)$ states in any graph with minimum degree at least $\delta$ and maximum degree at most $\Delta$.

For regular expander graphs, this matches the optimal space complexity of $\Theta(\log n)$ for fast protocols in complete graphs [Alistarh et al., SODA 2016 and Doty et al., FOCS 2022] with a nearly optimal stabilization time of $O(n \log^2 n)$ steps.
Finally, we give a new upper bound of $O(\trel \cdot n \log n)$ for the stabilization time of a constant-state protocol.

\end{myabstract}

\thispagestyle{empty}
\setcounter{page}{0}
\newpage

\section{Introduction}
Population protocols model distributed computation in networks of resource-constrained, anonymous agents that interact asynchronously~\cite{angluin2006computation}.
In this model, the agents are represented by nodes of a connected graph $G = (V,E)$, and in each time step, a pair of adjacent nodes are chosen to interact.
The fundamental question in this area is what can be computed -- and how fast -- by simple protocols that use only a small number of states per agent~\cite{angluin2006computation,AAER07,angluin2008simple,angluin2008fast}.
Time complexity is measured as \emph{expected time} under a stochastic scheduler that selects interacting pairs of nodes uniformly at random.

The model has become popular for investigating the limits of molecular and biological computation.
Most of the research has focused on the special case where the interaction graph $G$ is complete~\cite{AAER07,angluin2008simple,doty_stable_2018,alistarh2015fast,alistarh2017time,berenbrink2018population,berenbrink2021time,berenbrink_optimal_2020}, as this models well-mixed chemical solutions and is a special case of the chemical reaction network formalism.
In comparison, we focus on the more challenging general setting of arbitrary interaction graphs, which has so far received much less attention~\cite{draief2012convergence,mertzios2017determining,alistarh2021fast,alistarh2025near}.%

\paragraph{Majority consensus.}
We give new upper and lower bounds for the \emph{exact majority consensus} task, which is a fundamental agreement problem in distributed computing~\cite{AAER07,alistarh2015fast,draief2012convergence,benezit2009interval,mertzios2017determining}.
In this problem, each node $v \in V$ is given a binary input $f(v) \in \{0,1\}$ and all nodes should eventually agree (i.e., reach consensus) on the input given to the majority of the nodes.

The \emph{exact} variant of the majority consensus problem requires
that the nodes reach consensus on the majority value almost surely (with probability 1).
In contrast, the easier \emph{approximate} version asks for consensus on the majority value only with high probability, and only when the initial bias between the opinions is sufficiently large~\cite{angluin2008simple}; the latter variant tends to be much easier to solve than the exact variant~\cite{alistarh2018space,doty2022time}. In this work, we focus on the harder exact version.

\paragraph{State-of-the-art: Complete interaction networks.}
By now, the complexity of majority consensus is essentially settled in complete interaction graphs.
The approximate version can be solved with only 3-states in optimal $\Theta(n \log n)$ time~\cite{angluin2008simple}.
In contrast, exact majority consensus has been shown to exhibit space-time complexity trade-offs~\cite{alistarh2015fast,alistarh2017time,berenbrink2018population,berenbrink2021time,doty2022time,mertzios2017determining}.
Simply put, for exact majority consensus, constant-state protocols cannot be fast~\cite{alistarh2017time,alistarh2018space}.
On the other hand, protocols using $\Theta(\log n)$ states can achieve a near-linear speedup and stabilize in optimal $\Theta(n \log n)$ time~\cite{doty2022time}.

\paragraph{Our focus: general graphs.}
Unlike for complete graphs, complexity trade-offs for exact majority remains wide open in general interaction graphs.
There exists a simple 4-state protocol with polynomial time complexity~\cite{benezit2009interval,draief2012convergence,mertzios2017determining}. There are also protocols that are often (but not always) faster, but they use \emph{polynomially} many states in the worst-case~\cite{berenbrink2016plurality,alistarh2021fast}.
It is not known if these protocols are time or space optimal, apart from the space optimality of the 4-state protocol~\cite{mertzios2017determining}.

We prove the first time lower bounds for exact majority in general graphs, which are tight for general (unbounded space) protocols.
We design a fast protocol
that uses only \emph{polylogarithmic} number of states on any graph.
On graphs with good expansion, our protocol has near-optimal time complexity and its space complexity matches the state-of-the-art protocol in complete graphs~\cite{doty2022time}.
On graphs with low expansion, our work gives an \emph{exponential} improvement in the state-of-the-art space complexity~\cite{alistarh2021fast} while still achieving a better time complexity.
Finally, we give a new bound for the 4-state protocol, showing that it is at most a linear factor slower than our fast protocol.

\subsection{Population protocols and majority consensus}

Let $G = (V,E)$ be a simple graph with $n$ nodes and $m$ edges.
All nodes are anonymous, i.e., they do not know their identity or degree in the graph $G$.
In addition, for the majority problem, each node is given an input bit $f(v) \in \{0,1\}$.
The computation proceeds asynchronously in discrete time steps.
Here, we assume a stochastic scheduler, where in each time step $t \ge 1$,
\begin{enumerate}[noitemsep]
\item the scheduler samples uniformly at random a pair $e_t = (u,v)$ of adjacent nodes, and %
\item nodes $u$ and $v$ exchange information and update their states.
\end{enumerate}
In addition, at each time step $t$, each node $v$ maps its current state to an output value $g_t(v) \in \{0,1\}$.

The (normalized) \emph{bias} of the input is given by $\gamma := | |f^{-1}(0)|-|f^{-1}(1)||/n$, that is, the difference between the majority and minority input values divided by $n$.
Throughout, we assume that the input satisfies $\gamma>0$ (i.e., one of the values is a true majority value). Note that $1/n \le \gamma \le 1$.

\paragraph{Configurations and executions.}
Formally, a \emph{configuration} is a map~$x \colon V \to \Lambda$, where $\Lambda$ is the set of states a node can take.
In Step (2), the selected nodes update their states by applying the \emph{state transition function} $\Xi :  \Lambda \times \Lambda \rightarrow \Lambda \times \Lambda$ of the protocol. We often describe the protocols using rules of the form
\[
a + b \to c + d,
\]
which denote that $\Xi(a,b) = (c,d)$.
For any two configurations $x$ and $y$, we write $x \Rightarrow_{(u,v)} y$ if
\[
(y(u), y(v)) = \Xi(x(u), x(v))
\]
and $x(w)=y(w)$ for all $w \notin \{u,v\}$.
A configuration $y$ is \emph{reachable from} configuration $x \neq y$ if there exists some $k \ge 1$ and $(e_1, \ldots, e_k)$ such that $x \Rightarrow_{e_1} x_1  \Rightarrow_{e_2} \cdots \Rightarrow_{e_k}   y$.

The \emph{schedule} is the infinite random sequence $\sigma = (e_t)_{t \ge 0}$.
If $e_t = (u,v)$, we say that $u$ is the \emph{initiator} and $v$ is the \emph{responder} at time step $t$.
An \emph{execution} is the random sequence $(X_t)_{t \ge 0}$ of configurations, where $X_0$ is the initial configuration and $X_t \Rightarrow_{e_{t+1}} X_{t+1}$ for all $t \ge 0$.
Without loss of generality, we assume that the initial configuration $X_0$ is given by the input $f \colon V \to \{0,1\}$.

\paragraph{Stabilization and space complexity.}
A configuration $x$ is \emph{stable} if for all configurations $y$ reachable from $x$, the output of every node is the same in $x$ and $y$. A configuration $x$ has \emph{reached consensus} if all outputs are the same.
Furthermore, the configuration $x$ has \emph{reached majority consensus} if all outputs are equal to
the initial majority value $b = \argmax |f^{-1}(b)|$.
We are interested in the \emph{stabilization time} of protocols, which is given by the random variable
\[
T = \min \{ t : X_t \textrm{ is a stable configuration} \}.
\]
The \emph{time complexity} of a majority consensus protocol is the expected stabilization time $\Exp[T]$.
We say that $T \in O(h(n))$ holds \emph{with high probability} if for any constant~$\kappa>0$ there exists a constant~$c(\kappa)$ such that $\Pr[T \ge c(\kappa) h(n)] \le 1/n^\kappa$.

Finally, the \emph{space complexity} of a protocol is the number $|\Lambda|$ of states.
It is important to note that this is \emph{not} the number of \emph{bits} needed to encode the state.

\subsection{Population random walks and relaxation times}

Our time and space complexity bounds will be parameterized by the spectral properties of a (lazy) random walk induced by the stochastic scheduler.
The (lazy) \emph{population random walk} on $G = (V,E)$ is the Markov chain $Y = (Y_t)_{t \ge 0}$ on state space $V$, where $Y_t \in V$ gives the location of the random walker at time step $t$. The state transition matrix $P = (p_{uv})_{u,v \in V}$ of $Y$ is given by
\[
p_{u,v} = \begin{cases}
  \frac{1}{2m} & \textrm{if $\{u,v\} \in E$}, \\
  1-\frac{\deg(u)}{2m} & \textrm{if $u=v$}, \\
  0 & \textrm{otherwise,}
  \end{cases}
\]
where $\deg(u)$ denotes the degree of a node $u \in V$. Here $1/(2m)$ is the probability that a scheduler samples a specific pair $(u,v)$ of adjacent nodes to interact.
Intuitively, this random walk captures how a single token moves in token-based population protocols~\cite{alistarh2025near,berenbrink2016plurality,sudo2021self} and in related stochastic processes, where edges are sampled uniformly at random~\cite{caputo2010proof,jonasson2012interchange,morris2008spectral,oliveira2013mixing,hermon2020exclusion}.
Note that the population random walk is different from the \emph{classic lazy random walk} on $G$, where the random walker stays put with probability $1/2$, and otherwise, the walker moves to a neighbor selected uniformly at random.

\paragraph{Spectral gap and relaxation time.}
The population random walk $Y$ is irreducible, aperiodic and reversible. Hence, the spectrum of $P$ is positive and the eigenvalues of $P$ satisfy
\[
\lambda_1 = 1 > \lambda_2 \ge \cdots \ge \lambda_n \ge 0.
\]
The \emph{spectral gap} of $P$ is $1 - \lambda_2$. The \emph{relaxation time} of the random walk $Y$ is $\trel := 1/(1-\lambda_2)$.
The unique stationary distribution of $P$ is the uniform distribution given by $\pi(x)=1/n$; see, e.g.,~\cite{alistarh2021fast}.

The relaxation time of a random walk is closely related to its \emph{mixing time}. However, the relaxation time can sometimes be smaller than the mixing time. In particular,
the mixing time $\tmix$ of the random walk $Y$  lies between $\Omega(\trel)$ and $O(\trel \log n)$~\cite[Theorem 12.3 and Theorem 12.4]{levin2017markov}.

\paragraph{Population vs.\ classic random walks.}
One may wonder why we use the relaxation time of \emph{population} random walks instead of the more widely studied relaxation time of \emph{classic} random walks. The simple answer is that the population random walk more accurately expresses how tokens in our protocols move and that these two random walks behave somewhat differently in general.

For regular graphs,
the state transition matrix $P$ of the lazy \emph{population} random walk and the state transition matrix $P'$ of the lazy \emph{classic} walk satisfy the convenient equality $(P-I) =(2/n)(P'-I)$.
Thus, in particular, the relaxation time of the population random walk is the relaxation time of the classic random walk scaled by a factor of $\Theta(n)$.
This means that -- in regular graphs -- the relaxation time $\trel$ of the population random walk appearing in our bounds can be simply replaced with a $\Theta(n\trel')$ factor, where $\trel'$ is the relaxation time of a classic random walk.

In general graphs, a simple scaling by a $\Theta(n)$ factor does not work. To see why, consider a lollipop graph, i.e., a clique of size $\Theta(n)$ attached to a path of length $\Theta(n)$. In a lollipop graph, the relaxation time of the classic random walk is $O(n^2)$, but the population random walk has relaxation time of $\Omega(n^4)$. We are not aware of a simple way to express the relaxation time of a population random walk with the relaxation time of a classic walk in general.

\subsection{Our contributions}

We give new upper and lower bounds for the time and space complexity of exact majority on general graphs.
Our results are summarized and compared to state-of-the-art in \Cref{table:summary}.

\begin{table}[t]
\begin{center}
\begin{tabular}{@{}llll@{}}
\toprule
Graph & Stabilization time & Number of states & Reference \\
\midrule
Cliques &  $O(n^2 \log n / \gamma)$ & $O(1)$ & \cite{benezit2009interval,draief2012convergence,mertzios2014determining-conference-version} \\
 & $O(n \log n)$ & $O(\log n)$ & \cite{doty2022time}  \\
 & $O(n^{2-\varepsilon})$ & $\Omega(\log n)$ & For any constant $\varepsilon>0$~\cite{alistarh2018space}  (*) \\
 &  $\Omega(n \log n)$ & any &  Folklore (coupon collector) \\
\midrule
Regular & $O(\phi^{-2} n \log^6 n)$ & $O(\phi^{-2} \log^5 n)$ & \cite{alistarh2021fast} \\
 & $O(\trel \log^2 n)$ & $O(\log n \cdot \log(\trel/n))$ &  \textbf{This work: \Cref{corollary:fast_regular}} \\
 & $\Omega( \max\{D, \log n\} \cdot n)$ & any & \textbf{This work: \Cref{thm:lower_bound}} \\
\midrule
Connected & $O(n^6)$ & $O(1)$ & Analysis of~\cite{mertzios2014determining-conference-version} \\
 & $O(\log n / \delta(\gamma) )$ & $O(1)$ & Analysis of \cite{draief2012convergence} (**)  \\
 & $O( \trel \log n / \gamma )$ & $O(1)$ & \textbf{This work: \Cref{thm:constant_state_upper_bound}} \\
 & $O((\Delta/\delta) \trel \log^2 n)$ & {\small $O\left( \log n \cdot \left( \log \frac{\Delta}{\delta} + \log \frac{\trel}{n} \right) \right)$} & \textbf{This work: \Cref{thm:fast_upper_bound}} \\
 & $\Omega( Dm )$  & any & \textbf{This work: \Cref{thm:general_lower_bound}} (***)\\
 \bottomrule
\end{tabular}
\end{center}
\caption{Complexity bounds for exact majority in different graph families. Here $1/n \le \gamma \le 1$ is the initial normalized bias between majority and minority opinions, $\phi$ is the conductance of a regular graph, and $\trel$ is the relaxation time of the lazy population random walk. The time bound gives both the expected and the w.h.p.\ bound for stabilization.
The $O(1)$-state protocol is always the 4-state protocol of B\'en\'ezit et al.~\cite{benezit2009interval}. (*) This lower bound of Alistarh et al.~\cite{alistarh2018space} is conditional and holds only for protocols that are ``output dominant'' and ``monotone''. (**) The parameter $\delta(\gamma)$ in the bound of Draief and Vojnovi\'c~\cite{draief2012convergence}, which is the minimum spectral gap over a certain family of submatrices. (***) This lower bound holds for worst-case graphs, not all graphs.}
\label{table:summary}
\end{table}

\paragraph{Time lower bounds on general graphs.}
We give new unconditional time lower bounds for exact majority on graphs.
The lower bounds hold even if (a) the protocol uses an unbounded number of states, (b) nodes have access to an infinite stream of random bits, and (c) nodes know the topology of the interaction graph $G$.
These time lower bounds are asymptotically tight if we do not restrict the number of states used by the nodes.
Our first result holds for any regular graph.

\begin{restatable}{theorem}{regularlb}
\label{thm:lower_bound}
    For any regular graph $G$ with diameter $D$, the worst-case expected stabilization time is %
    \begin{equation*}
		\Omega( \max\{D, \log n\} \cdot n ).
    \end{equation*}
\end{restatable}

The same technique can be easily extended to show a worst-case bound for specific graphs.
For example, we show that on worst-case general graphs, the expected stabilization time can be a factor $n$ larger compared to regular graphs.

\begin{restatable}{theorem}{generallb}
\label{thm:general_lower_bound}
For every~$D(n) \in \{\log n, \dots , n\}$ and $m(n) \in \{n, \dots , n^2\}$, there exists a graph~$G$ with~$n$ nodes, $\Theta(m)$ edges and diameter~$\Theta(D)$ on which the worst-case stabilization time is $\Omega(Dm)$.
\end{restatable}

\paragraph{Fast space-efficient protocols.}
We give a new fast super-constant state protocol for general graphs whose running time and space complexity depend on the relaxation time $\trel$ of the population random walk and the
ratio $\Delta/\delta$, where $\delta$ is a lower bound on the minimum degree of $G$ and $\Delta$ is an upper bound on the maximum degree. The ratio measures how close to the graph is to being regular.

\begin{restatable}{theorem}{fastubgeneral}
\label{thm:fast_upper_bound}
For graphs with minimum degree $\delta>0$ and maximum degree $\Delta \le n$, there exists a protocol that uses $S$ states and stabilizes in $T$ steps in expectation and with high probability, where
\[
S \in O\left( \log n \cdot \left( \log \frac{\Delta}{\delta} + \log \frac{\trel}{n} \right) \right) \qquad \textrm{ and } \qquad  T \in O\left( \frac{\Delta}{\delta} \cdot \trel \log n \cdot \log(1/\gamma) \right).
\]
\end{restatable}

The general bound is a bit unwieldy, but for regular graphs it simplifies to the following.

\begin{corollary}\label{corollary:fast_regular}
For regular graphs, there exists a protocol that uses $S$ states and stabilizes in $T$ steps in expectation and with high probability, where
\[
 S \in O\left( \log n \cdot \log \frac{\trel}{n} \right) \qquad \textrm{ and } \qquad  T \in O\left( \trel \log n \cdot \log(1/\gamma)  \right).
\]
\end{corollary}
We note that constant-degree regular expander graphs and random regular graphs (w.h.p.) have $\trel \in O(n)$. For such graphs, the space complexity becomes $O(\log n)$ while the stabilization time is near-optimal $O(n \log^2 n)$. This asymptotically matches the state-of-the-art space complexity and nearly matches the time complexity (up to a $O(\log n)$ factor) in complete graphs~\cite{doty2022time}, which are expanders with diameter one, whereas, e.g., constant-degree expanders have diameter $\Theta(\log n)$.

\paragraph{New stabilization time bound for the constant-state protocol.}
B\'en\'ezit et al.~\cite{benezit2009interval} designed  a simple 4-state protocol that solves exact majority on any connected graph.
Draief and Vojnovi\'c~\cite{draief2012convergence} gave a clever spectral analysis of the (continuous) stabilization time of this protocol in terms of the minimum spectral gap $\delta(\gamma)$ of a \emph{family of matrices} derived from the continuous-time interaction rate matrix $Q$ (i.e., the generator of the continuous-time random walk on $G$). They gave closed-form bounds on $\delta(\gamma)$ only for specific graph classes, such as stars, cycles and Erd\H{o}s--R\'{e}nyi random graphs.

We show that the quantity $\delta(\gamma)$ can be replaced with the easier-to-interpret quantity $\gamma/\trel$, which can be derived from the spectral gap of the population random walk on $G$. This yields a simpler upper bound for the (discrete) stabilization time of the 4-state protocol.

\begin{restatable}{corollary}{constantub}
\label{thm:constant_state_upper_bound}
There exists a 4-state protocol that stabilizes in
$O(\trel \log n / \gamma)$
steps in expectation and with high probability when the input has bias $\gamma > 0$.
\end{restatable}

This shows that for expander graphs, the stabilization time is $O(n \log n/\gamma)$ in expectation and with high probability. Thus, in expanders, the 4-state protocol is time and space optimal whenever the bias $\gamma$ is at least constant.

\subsection{Overview of techniques}\label{ssec:overview}

We now give a brief overview of the main ideas for proving our results; for now, we keep the discussion relatively informal and swipe several technical details under the rug.

\paragraph{Lower bounds.}
Our time lower bounds are based on information propagation and indistinguishability arguments: a node needs to hear from sufficiently many nodes in order to output the correct majority value, as otherwise it cannot distinguish between two inputs with different majorities.
The main technical challenge is lower bounding the time this happens under  the stochastic scheduler.

For this, we use the notion of \emph{broadcast time}, similarly to the recent work of Alistarh et al.~\cite{alistarh2025near}. This is the time for an ``epidemic'' started from a single source node $v$ to reach the entire network. Alistarh et al.~\cite{alistarh2025near} gave bounds for the broadcast time from a single source node. %
Here, we slightly generalize this notion to broadcast time from a \emph{set} of nodes. With this, we can show that unless nodes have not heard from a specific set $A \subseteq V$, their output will be wrong with constant probability.

\paragraph{Fast exact majority via synchronized cancellation--doubling.}
Angluin et al.~\cite{angluin2008fast} introduced early on the idea of synchronized cancellation and doubling phases to solve exact majority fast in complete graphs. Still, essentially all fast exact majority protocols are built on this idea~\cite{angluin2008fast,elsaesser2018recent,alistarh2018recent,berenbrink2018population,berenbrink2018population,alistarh2021fast,doty2022time}.
Our protocol for arbitrary interaction graphs from \Cref{thm:fast_upper_bound} also uses synchronized cancellation-doubling to amplify the bias fast.

The high-level idea is that initially, a node has  a \emph{strong opinion token} $\atype$ if it has input 0 and a \emph{strong opinion token} $\btype$ if it has input 1. For now, assuming that all nodes would know the current time step $t$, two phases are  \emph{synchronously} alternated in an iterative fashion:
\begin{itemize}%
    \item In the \emph{cancellation phase}, tokens are updated using the rule $\atype + \btype \to \ctype + \ctype$ until
at least a fraction $9/10$ of the tokens are of type $\ctype$ \emph{or} there are no minority opinion tokens left.

    \item In the \emph{doubling phase}, tokens updates are with the rules $\atype + \ctype \to \smallatype + \smallatype$ and  $\btype + \ctype \to \smallbtype + \smallbtype$ until
there are no more $\atype$ and $\btype$ tokens \emph{or} there are no more $\ctype$-tokens.
If the latter happens, then (w.h.p.) there is only strong opinions of one type, and this type is declared as the majority.
\end{itemize}
At the end of the doubling phase, the rules $\smallatype \to \atype$ and $\smallbtype \to \btype$ are applied locally.
Iterating the above scheme for $O(\log 1/\gamma)$ iterations, where $\gamma>0$ is the initial bias between $\atype$ and $\btype$, suffices, as each iteration either doubles the bias between $\atype$ and $\btype$ or there are no more minority opinion tokens left.

\paragraph{The challenge: spatial dependencies.}
When $G$ is an arbitrary interaction graph, it is not so obvious how this scheme can be implemented in the (asynchronous) population protocol model. (Indeed, it is not so easy even in the case of the complete graphs~\cite{,alistarh2018space,alistarh2017time,doty2022time,berenbrink2018population,berenbrink2021time}.)
In particular, there are two questions that we need to address for arbitrary interaction graphs:
\begin{enumerate}[noitemsep]
	\item How do we bound the time needed to run the cancellation and doubling phases?
	\item How to synchronize the population in a time- and space-efficient manner?
\end{enumerate}
Unfortunately, we cannot just use existing arguments and techniques that work in complete graphs.
The key challenge is that -- unlike in complete graphs -- the probability of an interaction between nodes in state $x$ and $y$ at time $t$ is no longer simply proportional to the counts of states $x$ and $y$ at time $t$ in general graphs.
The spatial dependencies between local states make the analysis of many dynamics in graphs highly non-trivial.
To deal with these, we develop two new main ingredients.

\paragraph{Ingredient 1: New bounds for annihilation-diffusion dynamics on graphs.}
To resolve the first question, we analyze an \emph{annihilation-diffusion process} on graphs with two input species~$\atype$ and~$\btype$. The process is simply given by the rules of the cancellation phase  with
\begin{equation}
\label{eq:annihilation}
\begin{split}
	x + y \to \ctype + \ctype &\qquad \text{if } \{x , y\} = \{\atype , \btype\}, \\
	x + y \to y + x &\qquad \text{otherwise.}
\end{split}
\end{equation}
In this process,
the individuals of species $\atype$ and $\btype$
traverse the graph by performing a population random walk.
When individuals of different species meet, both are cleared from the system and the nodes they occupied become empty; this is represented by the token  $\ctype$.

The \emph{extinction time} of this process is the minimum time until one of the species $\atype$ or $\btype$ becomes extinct, i.e., their count hits zero.
We show that the extinction time is $O(\trel \log n/\gamma)$ with high probability.
Using a stochastic domination argument, we bound the \emph{$\varepsilon$-clearing time} of the process, which is the minimum time until extinction \emph{or} at most an $\varepsilon$-fraction of the nodes are non-empty.
We show that for any constant $\varepsilon>0$, the $\varepsilon$-clearing time is $O(\trel \log n)$ with high probability.

The $\varepsilon$-clearing time can be used to bound the time it takes for ``sufficiently many distinct tokens of type $x$ to meet with sufficiently many distinct tokens of type $y$''. This is often needed to analyze various population protocols, and we believe bounds on this time will be of independent interest.

Proving the extinction time bound turns out to be quite tricky because the tokens do \emph{not} make independent random walks and the number of tokens changes over time.
To deal with this, we adapt the approach of Draief and Vojnovi\'c~\cite{draief2012convergence} and analyze a differential equation describing the evolution of the expectation that a given node has a minority token over  time. In their analysis, the extinction time (corresponding to their Phase 1) is bounded by
a minimum eigenvalue of a \emph{family} of submatrices of the generator $Q$ of the continuous-time random walk on $G$.

While we take exactly the same approach, we make an important key change: we consider different submatrices in the analysis, which will have a more convenient form.
This allows us to invoke results on mean \emph{quasistationary mean exit times} of Markov chains~\cite{collet2012quasi,aldous-fill-2014}. Specifically, we can connect extinction time to the spectrum of principal submatrices of the generator of the random walk.
Finally, a standard concentration bound for Poisson random variables allows us to translate the continuous-time bound to a discrete time bound.
The details are given in \Cref{sec:cancellation}.

As a by-product this also gives us the simpler-to-state upper bound for the stabilization time of the 4-state protocol given in \Cref{thm:constant_state_upper_bound}, by retracing the steps of Draief and Vojnovi\'c~\cite{draief2012convergence}.

\paragraph{Ingredient 2: Space-efficient graphical clocks.}
To address the second question of how to synchronize the population, we devise a new space-efficient clock for general graphs.
Compared to the graphical phase clock of Alistarh et al.~\cite{alistarh2021fast} for regular graphs, our clock is guaranteed to work in general graphs and has better complexity guarantees. While the space complexity of their clock degrades \emph{quadratically} in the conductance of the graph, ours scales \emph{logarithmically} in the relaxation time $\trel$ and the degree imbalance $\Delta/\delta$.

To further save space, we  divide the tokens in the population into \emph{clock tokens} and \emph{opinion tokens}. The former are used to run the clock, the latter to run the cancellation--doubling dynamics.
Specifically, the clock tokens driving our phase clock will use
\[
O\left( \log n \cdot \left( \log \left( \frac{\Delta}{\delta }\right) + \log\left( \frac{ \trel }{ n} \right) \right) \right)
\]
states, and nodes not holding a clock token need $O(1)$ states for the phase clock.
The disjoint set of opinion tokens run the iterative cancellation--doubling rules and use $O(\log n)$ states to count the number of iterations.

To get a globally synchronized phase clock, we start with a collection of \emph{active} clock tokens.
When some active clock token generates a clock tick, this triggers a propagating ``wave front'' that increments the phases of all nodes as it reaches them. Clock tokens that are too slow become \emph{deactivated} by this wave-front, and they cannot trigger phase increments anymore. This ensures that (w.h.p.) the phase clock remains tightly synchronized across the population.

\subsection{Related work}

\paragraph{Exact majority in complete graphs.}
The constant-state protocol for arbitrary interaction graphs solves exact majority in $O(n^2 \log n)$ expected time in complete graphs~\cite{draief2012convergence,benezit2009interval}. Alistarh et al.~\cite{alistarh2017time} showed that -- for a large class of protocols -- any protocol that stabilizes in $O(n^{2-\varepsilon})$ expected steps requires $\Omega(\log n)$ states, for any constant $\varepsilon > 0$.
A simple coupon collector argument shows that \emph{any} protocol for exact majority requires $\Omega(n \log n)$ time.
Both lower bounds were recently matched by the $\Theta(\log n)$-state protocol of Doty et al.~\cite{doty2022time} with the optimal stabilization time of~$\Theta(n \log n)$.

\paragraph{Constant-state exact majority in general interaction graphs.}
The constant-state protocol of B{\'e}n{\'e}zit et al.~\cite{benezit2009interval} also solves exact majority on any connected graph.
Draief and Vojnovi{\'c} \cite{draief2012convergence} analyzed the stabilization time in the continuous time Poisson clock model. However, they only gave closed-form solutions for selected graph families. Mertzios et al.~\cite{mertzios2017determining} also analyzed the same 4-state protocol and gave a $O(n^6)$ worst-case expected time bound for any connected graph.

In comparison to the previous stabilization time bounds for the 4-state protocol, our new analysis of the annihilation-diffusion process yields a more refined and easily intrepretable bound, as relaxation time can be directly connected to the mixing time, spectral gap and expansion properties of the graph.
 For example, our bound implies that the protocol stabilizes in $O((m/\zeta)^2 \log n)$ time on graphs with $m$ edges and edge expansion $\zeta>0$.

\paragraph{Fast exact majority in graphs.}
For general graphs, there are fewer protocols that trade space for speed.
Berenbrink et al.~\cite{berenbrink2016plurality} gave protocols for plurality consensus, a generalization of majority consensus to multiple values. For majority consensus, their protocols compute the majority with high probability in time depending linearly on the spectral gap of the interaction matrix. However, their protocols use \emph{polynomially many} states (i.e., logarithmically many bits). %

Alistarh et al.~\cite{alistarh2021fast} gave a protocol for $d$-regular graphs whose time and space complexity is parameterized by the conductance $\phi = \zeta/d$ of the graph, where $\zeta>0$ is the edge expansion. Their protocol stabilizes in $\phi^{-2} \cdot n \polylog n$ steps in expectation  and uses $\phi^{-2} \cdot \polylog n$ states. The quadratic dependency on $\phi$ arises from bounds of the mixing time of (a variant of) the interchange process, which is a token shuffling process on graphs~\cite{caputo2010proof,jonasson2012interchange}.
The space complexity mainly stems from
their phase clock construction, which leverages the $(1+\beta)$-balancing process of Peres, Talwar and Wieder~\cite{peres2015graphical}, whose
gap can be  polynomial in graphs with poor expansion.
For regular expanders graphs, the protocol still achieves polylogarithmic space complexity with $n \polylog n$ running time, but for poor expanders, such as cycles, the space bounds become polynomial.

\paragraph{Space-time trade-offs for leader election.}
Leader election is another central problem in the population protocol model. Its progress has closely paralleled that of exact majority. In complete graphs, the problem exhibits space-time trade-offs:
Doty and Soloveichik~\cite{doty_stable_2018} developed a surgery technique to show an \emph{unconditional} and tight time lower bound of $\Omega(n^2)$ for constant-state protocols. %

Alistarh et al.~\cite{alistarh2017time} extended the surgery technique to show that
when the number of states remains below~$\frac{1}{2} \log \log n$, any leader election protocol requires $n^2 /\polylog n$ time in complete graphs.
Sudo and Masuzawa~\cite{sudo_leader_2020} showed that for any (unbounded space) protocol requires $\Omega(n \log n)$ steps to solve leader election in complete graphs.
These lower bounds were matched by the protocol of Berenbrink et al.~\cite{berenbrink_optimal_2020} that uses $\Theta(\log \log n)$ states and stabilizes in $\Theta(n \log n)$ expected steps in complete graphs.

For general graphs, Beauquier et al.~\cite{beauquier2013self} gave a 6-state leader election protocol that works in any connected graphs. The stabilization time was later shown to be $O(H(G) n \log n)$, where $H(G)$ is the worst-case hitting time of a classic random walk~\cite{sudo2021self,alistarh2025near}.
Recently, Alistarh et al.~\cite{alistarh2025near} gave polylogarithmic state protocols for leader election in general graphs. They showed that $O(\log^2n)$ states is enough to solve the problem in near-optimal $O(m/\zeta \log^2 n)$ time.

\paragraph{Exclusion and interchange process.}
Another approach for dealing with spatial dependencies among token dynamics is to leverage mixing time bounds for the \emph{interchange process} and the \emph{exclusion process} on graphs~\cite{caputo2010proof,jonasson2012interchange,morris2008spectral,oliveira2013mixing,hermon2020exclusion}. In these processes, each node starts with a token, and when edges are sampled, the tokens at the end points of sampled edges are swapped. In the interchange process, each token has a unique identifier, whereas in the exclusion process tokens have colours and tokens of the same colour are indistinguishable from one another.

After these processes mix, the tokens are close to a uniform distribution, so node states have low spatial correlation. This is used by the  majority protocol of Alistarh et al.~\cite{alistarh2021fast} for regular graphs. The protocol repeatedly runs the interchange process until its mixing time, and then simulates a single round of a synchronous protocol as if in a complete graph.
However, obtaining tight mixing time bounds for the interchange and exclusion process is highly non-trivial and are only known for selected graph families~\cite{morris2008spectral,jonasson2012interchange,oliveira2013mixing,hermon2020exclusion}.

We deal with the spatial dependencies using the annihilation--diffusion process; this simplifies the analysis and speeds up the process, as we do not have to wait for the tokens to mix and ``forget their starting points''. We are not worried about the locations of the tokens, but rather only that there have been sufficiently many interactions between different token types.

\paragraph{Annihilation--diffusion dynamics and coalescing random walks on graphs.}
Annihilation--diffusion dynamics similar to ours, where particles randomly move in space and annihilate each other upon contact, have been studied in the area of interacting particle systems~\cite{bramson1991asymptotic,ahlberg2024annihilating,bahl2022diffusion,ahlberg2021fixate}
The work in this area differs, as these dynamics are often considered on infinite lattices and graphs, often with random initial configurations.
In particular, to our knowledge, there exist no similar bounds for the worst-case extinction time in arbitrary finite graphs.

A related class of dynamics are \emph{coalescing random walks}~\cite{cooper2013coalescing,canade2023coalescence,berenbrin2018tight}. Here, the difference is that upon meeting, the tokens merge upon meeting rather than annihilate. Furthermore, it is typically assumed that these random walkers are independent, unlike in our model.

\subsection{Outline of the paper}
We start with some preliminaries in \Cref{sec:preliminaries}. \Cref{sec:lower} gives our lower bound results.
\Cref{sec:cancellation} is dedicated to describing and studying the annihilation--diffusion dynamics on graphs, to which the analysis of the cancellation and doubling dynamics can be reduced. This analysis can also be used to establish the new running time bound of the constant-state protocol.
\Cref{sec:clocks} describes the construction and analysis of our phase clock  used to run synchronized cancellation--doubling dynamics.
Finally, \Cref{sec:main_proof} combines the phase clocks and the annihilation--diffusion dynamics bounds to obtain our new fast protocol.

\section{Preliminaries}\label{sec:preliminaries}

\paragraph{Tail bounds for geometric random variables.}
Janson~\cite[Theorem 2.1 and Theorem 3.1]{janson2018tail} gave the following tail bounds for sums of geometric random variables.

\begin{lemma}\label{lemma:sum-of-geometric}
  Let $p_1, \ldots, p_k \in (0,1]$ and $X = Y_1 + \ldots + Y_k$ be a sum of independent geometric random variables with $Y_i \sim \Geom(p_i)$. Define $p = \min \{ p_i : 1 \le i \le k \}$ and $c(\lambda) = \lambda - 1 - \ln \lambda$. Then
  \begin{enumerate}[noitemsep,label=(\alph*)]
  \item $\Pr[ X \ge \lambda \cdot \Exp[X]] \le \exp\left(- p \cdot \Exp[X] \cdot c(\lambda)\right)$ for any $\lambda \ge 1$, and

  \item  $\Pr[ X \le \lambda \cdot \Exp[X]] \le \exp\left(- p \cdot \Exp[X] \cdot c(\lambda)\right)$ for any $0 < \lambda \le 1$.
  \end{enumerate}
  In particular, if all~$(p_i)_{1\le i \le k}$ are equal, then for any~$\lambda \geq 1$,
  \begin{enumerate}[resume,label=(\alph*)]
    \item $\P{ X \in \left[ \tfrac{1}{\lambda} \Exp[X] , \lambda \Exp[X] \right] } \geq 1-2\,\exp\pa{-k \cdot c(1/\lambda)}$.
  \end{enumerate}
\end{lemma}
\begin{proof}[Proof of (c)]
  If all~$(p_i)_{1\le i \le k}$ are equal to~$p \in [0,1]$, then~$\Exp[X] = k/p$. Therefore, by applying items (a) and (b) to~$\lambda$ and~$1/\lambda$ respectively, and using the union bound, we obtain
  \begin{equation} \label{eq:first_union_bound}
    \P{ X \in \left[ \tfrac{1}{\lambda} \Exp[X] , \lambda \Exp[X] \right] } \geq 1-\exp(-k \cdot c(\lambda)) -\exp(-k \cdot c(1/\lambda)).
  \end{equation}
  Note that~$c(\lambda) - c(1/\lambda) = \lambda - 1/\lambda - 2 \ln \lambda =: f(\lambda)$. The function~$f$ is derivable on~$[1,+\infty)$ and its derivative is non-negative. Moreover, $f(1) = 0$. Therefore, $f$ is non-negative on~$[1,+\infty)$, which implies that~$c(\lambda) \geq c(1/\lambda)$ on that interval, and the statement follows from \Cref{eq:first_union_bound}.
\end{proof}

\paragraph{Relaxation time and expansion.}
Let $G=(V,E)$ be a graph. For a set $S \subseteq V$, the edge boundary of $S$ is $\partial S = \{ \{ u,v\} \in E : u \in S, v \notin S \}$.
The \emph{edge expansion} of the graph $G$ is
\[
\zeta := \min \left\{ \frac{|\partial S|}{|S|} : S \subseteq V, |S| \le n/2 \right \}.
\]
The next bound follows from well-known results stating that the spectral gap of a Markov chain
is bounded by its conductance~\cite[Theorem 13.14]{levin2017markov}. The proof of the next lemma is in \Cref{apx:spectral-sandwich}.

\begin{restatable}{lemma}{spectralsandwich}\label{lemma:spectral-gap-sandwich}
  Let $\trel$ be the relaxation time of the population random walk on $G$.  Then %
\[
\frac{m}{\zeta} \le \trel \le 8 \left(\frac{m}{\zeta} \right)^2.
\]
\end{restatable}

\paragraph{Broadcast time.}
Following Alistarh et al.~\cite{alistarh2025near}, we use the notion of broadcast time in graphs, which is the time for a ``epidemics'' to spread in the population when it is started from a fixed node.
For each node $v \in V$, we define the \textit{set of influencers} of~$v$ by $I_0(v)=\{v\}$, and for $t\geq 0$,
\begin{equation*}
	I_{t+1}(v) = \begin{cases}
    		I_t(v) \cup I_t(u) &\text{ if } e_{t+1}=(u,v) \text{ or } e_{t+1}=(v,u), \\			I_t(v) &\text{ otherwise.}
	\end{cases}
\end{equation*}
Let $T(u,v) := \min \{t, u \in I_t(v) \}$ be the first step in which node~$v$ is
influenced by node~$u$. The {\em broadcast time}~$\tbroadcast(u)$ from node~$u$ is the first step in which all nodes are influenced by~$u$, i.e.,
\begin{equation*}
	\tbroadcast(u) := \max_{v \in V} ~ \{ T(u,v) \}.
\end{equation*}
Alistarh et al.~\cite{alistarh2025near} showed that $\tbroadcast(u) \in O\left( m \log n / \zeta \right )$ with high probability. Together with \Cref{lemma:spectral-gap-sandwich}, we get the following bound in terms of relaxation time $\trel$. The proof is in \Cref{apx:broadcast-relax}

\begin{restatable}{lemma}{broadcastrelax}\label{lemma:broadcast-faster-than-relaxing}
  For every~$\alpha \ge 1$, there is a constant $C(\alpha)$ such that
\[
\Pr[\tbroadcast(v) \ge C(\alpha) \trel \log n] \le \frac{1}{n^{\alpha}}.
\]
\end{restatable}

\section{Time lower bound for general protocols}\label{sec:lower}

In this section, we give unconditional time lower bounds for exact majority consensus in general graphs. The lower bounds hold even under unbounded space protocols that have access to an infinite stream of random bits and know the topology of the interaction graph $G$. For such protocols, the lower bounds are also tight. We start with a lower bound for \emph{all} regular graphs.

\regularlb*

To prove the result, we need to slightly generalize the notion of broadcast time from a source \emph{node} to broadcast time from source \emph{sets}.
The broadcast time from a set~$A \subseteq V$ is defined as the first step in which all nodes are influenced by {\em at least one} node of~$A$ given by
\begin{equation*}
	\tbroadcast(A) := \max_{v \in V} ~ \min_{u \in A} ~ \{ T(u,v) \},
\end{equation*}
where $T(u,v)$ was the first time step in which node $v$ was influenced by node $u$.

The next lemma establishes an upper bound on the broadcast time of some set~$\{u_1,u_2\}$, where~$u_1$ and~$u_2$ are adjacent; it follows by adapting the arguments used by Alistarh et al.~\cite[Section 3.4]{alistarh2025near}.
We give the proof of \Cref{thm:individual_broadcast_proba_lb}  in \Cref{apx:broadcast-lb-lemma}.

\begin{restatable}{lemma}{tbrlemma}\label{thm:individual_broadcast_proba_lb}
    There exists a constant~$c_0>0$ such that for any~$d$-regular graph~$G$ with diameter $D$, there is an edge~$e = \{u_1,u_2\} \in E$ satisfying
    \begin{equation*}
    	\Pr[ \tbroadcast(\{u_1,u_2\}) \leq c_0 \cdot \max\{D,\log n\} \cdot n] \in O\pa{n^{-2}}.
	\end{equation*}
\end{restatable}

With the above lemma, we are ready to prove \Cref{thm:lower_bound}.

\begin{proof} [Proof of \Cref{thm:lower_bound}]
    Fix a protocol and let $\lout \colon \Lambda \to \{0,1\}$ be the function that maps the state of a node to its output.
    Consider~$c_0>0$ and~$\{u_1,u_2\} \in E$, obtained from applying \Cref{thm:individual_broadcast_proba_lb} to~$G$.
    Let $\tau = \tbroadcast(\{u_1,u_2\}) - 1$ be the last time step before all nodes are influenced by $u_1$ or $u_2$.
    Writing $\alpha_n := c_0 \cdot \max\{D,\log n\} \cdot n$, \Cref{thm:individual_broadcast_proba_lb} states that
    \begin{equation} \label{eq:applying_theorem}
        \Pr[ \tau < \alpha_n ] \in O(n^{-2}).
    \end{equation}
    Let $W \in V \setminus \{I_\tau(u_1) \cup I_\tau(u_2)\}$ be the unique random node that is neither influenced by $u_1$ nor $u_2$ at time step $\tau$.

Consider an input $f \colon V \to \{0,1\}$ that assigns input values to the nodes in $V \setminus \{u_1,u_2\}$ as evenly as possible. Since $W$ is not influenced by $u_1$ or $u_2$ at time step $\tau$, its output $\lout(x_\tau(W)) \in \{0,1\}$ is independent of the input values $f(u_1)$ and $f(u_2)$. Let $z \in \{0,1\}$ be such that
    \[ \Pr[\lout(x_\tau (W)) = z] \geq 1/2. \]
    By setting $f(u_1) = f(u_2) = 1-z$ the majority input value of $f$ becomes $1-z$. Let $T$ be the stabilization time of the protocol on the graph $G$ with input $f$. Then we have that
    \begin{equation} \label{eq:wrong_assignment}
        \Pr[T > \tau] \geq 1/2,
    \end{equation}
    because node $W$ outputs at time step $\tau$ value $z$ with probability at least $1/2$. Since $z$ is not the majority input value, the configuration $x_\tau$ is not correct with probability at least $1/2$.
    Finally,
    \begin{align*}
        \Pr [ T \geq \alpha_n] &\geq
        \Pr [ T \geq \tau \mid \tau \geq \alpha_n] \cdot \Pr [\tau \geq \alpha_n] \\
        &= \Pr [ T \geq \tau] - \Pr [ T \geq \tau \mid \tau < \alpha_n] \cdot \Pr [\tau < \alpha_n] & \text{(law of total probability)} \\
        &\geq \frac{1}{2} - O(n^{-2}) & \text{(by Eqs. \eqref{eq:applying_theorem} and \eqref{eq:wrong_assignment})} \\
        &\geq \frac{1}{4}. & \text{(for large enough~$n$)}
    \end{align*}
    Since~$T$ is a non-negative random variable, this implies that~$\Exp[T] \geq \alpha_n/4$.
    By definition, the worst-case expected stabilization time is at least $\Exp[T]$, which concludes the proof of \Cref{thm:lower_bound}.
\end{proof}

We next use the same argument to show that there exist non-regular graphs, where the lower bound can be higher.

\generallb*
\begin{proof}
Let~$D \in \{\log n, \ldots, n\}$ and~$m \in \{n, \ldots, n^2\}$.
Let~$H$ be a star on~$n-D/2$ nodes, and let~$u$ be the center node of~$H$. Since~$n-D/2 \geq n/2$, there exists a constant~$\alpha > 0$ such that $\binom{n-D/2}{2} \geq \alpha n^2$. Therefore, we can construct a graph~$H'$ by adding arbitrary edges to~$H$, until~$H'$ has at least~$\alpha m$ edges.
Now, consider the graph~$G$ obtained from~$H'$ by connecting the center node~$u$ to a path on~$D/2$ nodes. By construction, $G$ has exactly~$n$ nodes, at least~$D/2+\alpha m = \Theta(m)$ edges, and diameter~$1+D/2 = \Theta(D)$.

Let~$u_1, u_2$, be the two nodes of the path which are farthest from~$H'$, and $e = \{u_1, u_2\}$.
The broadcast time $\tbroadcast(e)$ stochastically dominates the sum $X=\sum_{i=1}^{D/2}Y_i$ of i.i.d.\ geometric random variables $Y_i \sim \text{Geom}(1/m)$, since we need to sample the $D/2$ edges of the path in a specific order to influence nodes in~$H'$.
Note that $\mathbb{E}[X]= D m/2$ and $D \geq \log n$.
\Cref{lemma:sum-of-geometric}(b) implies that $\Pr[\tbroadcast(e) \leq \lambda \cdot D m/2] \leq \exp(-c(\lambda) \cdot D/2)$ for every $0 < \lambda < 1$.
With an analogous argument as in the proof of \Cref{thm:lower_bound}, we get that there exists an input $f$ such that the stabilization time $T$ satisfies $\mathbb{E}[T] \geq 1 / 4 \cdot Dm/2 \in \Omega(Dm)$, which concludes the proof of \Cref{thm:general_lower_bound}.
\end{proof}

\paragraph{Tightness of the lower bounds.}
We remark that the lower bounds are tight for general protocols. %
This shows that the lower bounds cannot be improved without making additional restrictions on the computational power of the protocol (e.g., limit the number of states or access to random bits).

We now sketch the proof to see why. Consider a protocol where all nodes start running the constant-state protocol from \Cref{thm:constant_state_upper_bound}.
In parallel, the nodes first generate unique identifiers with high probability, and then, start broadcasting their identifier-input pair to all other nodes.
When a node has received $n$ distinct identifier values, it knows the input values of all other nodes and can, therefore, correctly output the majority input value.
Until a node has seen $n$ distinct identifier values, it just uses the output of the constant-state protocol, which is correct with probability 1.

The nodes can generate unique identifiers with probability at least $1-1/n^{\kappa}$ using $(\kappa + 2)\log n$ random bits (similarly to \cite[Lemma 17]{alistarh2025near}).
Since $\gamma>1/n$ for any input and $\trel \in \poly(n)$ by \Cref{lemma:spectral-gap-sandwich}, the constant-state protocol stabilizes in $\poly(n)$ expected steps on any input. 
Choosing a sufficiently large constant $\kappa$ and letting  $\mathcal{E}$ be the event that the identifiers are unique, the stabilization time $T$ of the combined protocol satisfies
\[
\E[T] = \E[T \mid \mathcal{E}] \cdot \Pr[\mathcal{E}] + \E[T \mid \overline{\mathcal{E}}] \cdot (1- \Pr[\mathcal{E}]) \le \E[T \mid \mathcal{E}] + 1 \in O\left( \max\left\{\log n + D\right\} \cdot m\right),
\]
where bound on  the first term $\E[T \mid \mathcal{E}]$ follows using the union bound over all nodes, the time for every identifier-input pair to reach all nodes is with high probability (and in expectation) is in $O\left( \max\left\{\log n, D\right\} \cdot m\right)$ as shown by Alistarh et al.~\cite[Lemma 7]{alistarh2025near}.
In particular, this matches the lower bound of \Cref{thm:lower_bound} in constant-degree regular graphs and the bound of \Cref{thm:general_lower_bound}.%

\section{Two-species annihilation dynamics on graphs} \label{sec:cancellation}

In this section, we consider the two-species \emph{annihilation dynamics}, which will later help bound the time until many tokens of distinct types interact within more complex dynamics. Specifically, the results obtained here will be reused as a black box in the analysis of our main protocol in \Cref{sec:main_proof}.

\subsection{Annihilation dynamics}
The process is defined as follows. We have three token types $\atype, \btype, \ctype$.
The token types $\atype$ and $\btype$ represent the two species, whereas $\ctype$ corresponds to an ``empty'' node.
Recall from \Cref{ssec:overview} that when two nodes in states~$x, y \in \{\atype, \btype, \ctype \}$ interact, their type is updated using the rules
\begin{equation*}
  \tag{\ref{eq:annihilation}}
\begin{split}
	x + y \to \ctype + \ctype &\qquad \text{if } \{x , y\} = \{\atype , \btype\}, \\
	x + y \to y + x &\qquad \text{otherwise.}
\end{split}
\end{equation*}
The first rule is the \emph{annihilation rule}, which reduces the number of $\atype$ and $\btype$ tokens by one, and the second one is a \emph{swap rule}, which only moves the tokens and does not change their counts.

An execution of this dynamics under the stochastic scheduler induces a Markov chain $X = (X_t)_{t \ge 0}$, where $X_t \colon V \to \{ \atype, \btype, \ctype \}$.
We use the short-hands $\calA_t = X_t^{-1}(\atype)$, $\calB_t = X_t^{-1}(\btype)$ and $\calC_t = X_t^{-1}(\ctype)$. Note that $(\calA_t, \calB_t, \calC_t)$ gives a partition of $V$ for each $t \ge 0$.
Without loss of generality, we assume $|\calA_0| \ge |\calB_0|$ throughout, as the roles of $\atype$ and $\btype$ are symmetric.

\subsection{Extinction time of the annihilation dynamics} \label{sec:extinction}

We define the \emph{extinction time} of the two-species annihilation dynamics $X = (X_t)_{t \ge 0}$ as the random time $\textinction(X)$ in which the minority species is no longer present, i.e.,
\[
\textinction(X) = \min \{ t : \calB_t = \emptyset \}.
\]
For the annihilation dynamics,
we define the {\em bias}
between the two input species in time~$t \ge 0$ as
\[
\gamma(X_t) := \frac{|\calA_t| - |\calB_t|}{n}.
\]
Note that by definition of the dynamics, we have $\gamma(X_t) = \gamma(X_0) =: \gamma$ for all $t \ge 0$.
In this section, we prove an upper bound on~$\textinction(X)$ that holds with high probability and in expectation.

\begin{restatable}{theorem}{annihilationub}
\label{thm:extinction-time} \
For any~$\kappa > 1$ and any initial configuration of the two-species annihilation dynamics with initial bias $\gamma > 0$, the extinction time satisfies
\[
\Pr\left[\textinction(X) > \frac{(\kappa+1) \trel \ln n}{\gamma} \right] \le 2/n^\kappa.
\]
\end{restatable}
We note that it is not hard to verify that the above implies that $\Exp[\textinction] \in O(\trel \ln n / \gamma)$.

\paragraph{A naive first attempt.}
An initial (naive) proof strategy for \Cref{thm:extinction-time} might go something like this:
there are always at least $\gamma n$ majority species tokens, so all we need to do is to show that a fixed minority species token hits a set of size $\gamma n$ quickly and is removed from the system. For the population random walk, the worst-case mean hitting time to a set of size $\gamma n$ is $O(\tmix / \gamma)$; this follows from the fact that mixing times are related to hitting times of large sets~\cite{griffiths2014tight}.

One could hope that by bounding the removal time of a fixed token with a suitable tail bound, and dealing with dependencies using the union bound, would almost give the result -- at least up to a logarithmic factor, as $\tmix \in O(\trel \log n)$.
Unfortunately, since the tokens do not make independent random walks, and the counts of tokens change over time, it is not so obvious how to formalize the above intuitive reasoning.

\paragraph{A successful attempt.}
To prove \Cref{thm:extinction-time}, we consider the equivalent continuous-time annihilation process.
Each edge $\{u,v\} \in E$ is equipped with a Poisson clock that rings at a rate $q_{u,v}>0$, and when the clock rings, the tokens at $u$ and $v$ interact.
Let $Y$ be the discrete-time population random walk with transition matrix $P$ and
 $Y'$ be the continuous-time population random walk generated by the matrix~$Q = (q_{u,v})_{u,v \in V}$, where we set 
\begin{equation*}
	q_{u,v} = \begin{cases} p_{u,v} & \text{if } u \neq v, \\
	 - \sum_{u \neq w} p_{u,w}  & \text{otherwise.} \end{cases}
\end{equation*}
The unique stationary distribution $\pi$ of the discrete-time population random walk~$Y$ is uniform, and~$\pi$ is also the stationary distribution of $Y'$, as $Y'\pi=0$.
Since $P$ is a stochastic matrix with eigenvalues
\[
\lambda_1 = 1 > \lambda_2 \ge \cdots \ge \lambda_n \ge 0
\]
and the eigenvalues of $-Q=I_n-P$ are $\lambda'_i = 1-\lambda_i$ for each $1 \le i \le n$. Thus, the eigenvalues of $-Q$ are non-negative, and conversely, the eigenvalues of $Q$ are non-positive.
Finally, the relaxation time for both the discrete- and continuous-time population random walks is given by
\[
\trel = \frac{1}{1-\lambda_2} = \frac{1}{\lambda'_2}.
\]
For a symmetric square matrix $M$, we write $\lambda(M)$ for the \emph{largest} eigenvalue of $M$, i.e. $\lambda(M) = \max\{\lambda\mid\lambda\in\spec(M)\}$, where $\spec(M)$ is the set of eigenvalues of $M$.
For any set $U \subseteq V$, we write $Q[U]$ for the restriction of $Q$ to the set $U$, i.e., the principal submatrix of $M$ where we keep only the columns and rows whose index is in the set $U$.

Fix a set $S \subseteq V$; this will later correspond to the positions of the $\atype$ tokens.
We define the matrix $R_S = (r_{u,v})_{u,v \in V}$ as
\begin{equation} \label{eq:RS_def}
	r_{u,v} := \begin{cases}
 		q_{u,v} & \textrm{if } u,v \in V \setminus S \\
		\lambda(Q[V \setminus S]) & \textrm{if }  u=v \in S \\
		0 & \textrm{otherwise.}
	\end{cases}
\end{equation}
We remark that this matrix is different, but related to the matrix used by Draief and Vojnovi\'c~\cite[Eq. (4.1)]{draief2012convergence}.
In particular, $R_S$ will be block diagonal, which will be useful in simplifying the analysis.

We now diverge further from the analysis of Draief and Vojnovi\'{c} and connect the extinction time to quasistationary mean exit times of the random walk generated by $Q$.
The next lemma follows from the theory of quasistationarity distributions of Markov chains~\cite{collet2012quasi,aldous-fill-2014}; see e.g. the book by Aldous and Fill~\cite[Corollary 3.34, Eq.(3.83)]{aldous-fill-2014}.

\begin{lemma}\label{lemma:quasistationary-mean-hitting-time}
Let $S \subseteq V$. %
If the subgraph of~$G$ induced by~$V\setminus S$ is connected, then
\[
\frac{1}{-\lambda(Q[V \setminus S])} \le \frac{\trel}{\pi(S)},
\]
where~$\pi(S) = |S|/n$ is the weight of~$S$ w.r.t. the stationary distribution.
\end{lemma}

We remark that the quantity $-1/\lambda(Q[V \setminus S])$ can also be upper bounded by the worst mean hitting time to the set $S$; see~\cite[Eq.(3.83)]{aldous-fill-2014}.
Thus, in all our upper bounds, we could replace the relaxation time $\trel$ with the worst case hitting time to a set of size $\gamma n$. This matches the intuition that the extinction time is bounded by the time for a minority token $\btype$ to hit a large set of size $\gamma n$ made up of $\atype$-tokens, which always must exist at any point in time, as the annihilation dynamics preserves the bias. However, using the relaxation time often leads to tighter bounds.

\begin{figure}[t] %
  \begin{center}
\includegraphics[width=0.6\linewidth]{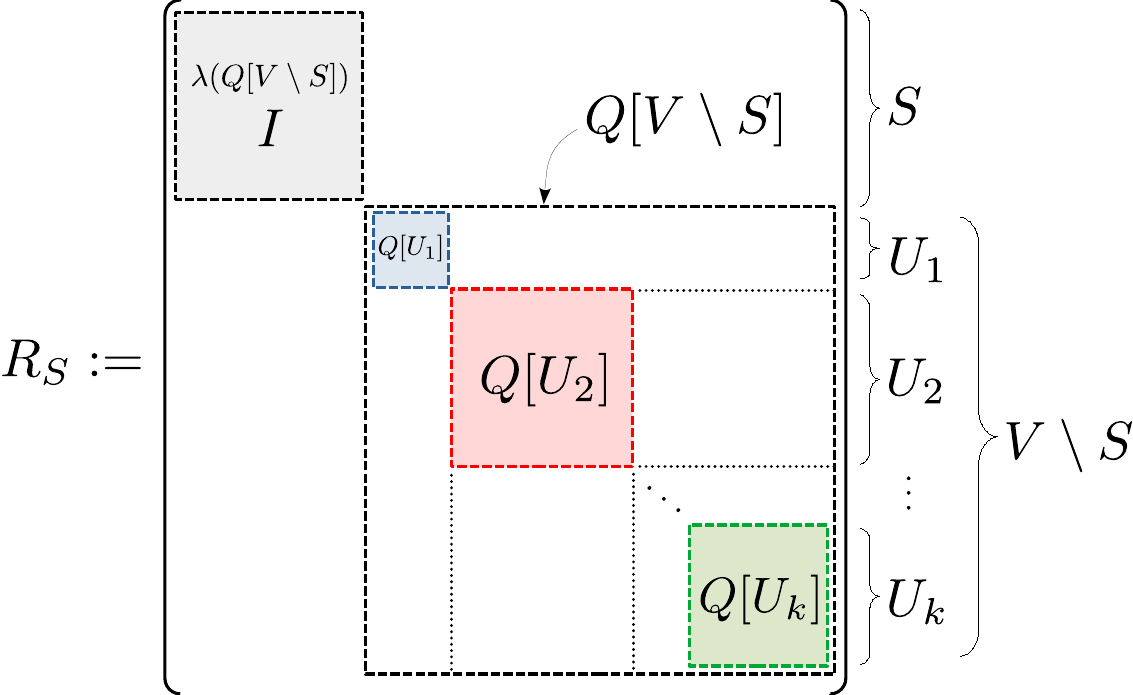}
\end{center}
\caption{Illustration of the matrix~$R_S$. All coefficients in colorless areas are equal to~$0$. The submatrix~$Q[V \setminus S]$, corresponding to the restriction of~$Q$ to~$V\setminus S$, is block diagonal since the sets~$U_1, \ldots, U_k$ are the connected components of the subgraph induced by $V \setminus S$.}
  \label{fig:rs-construction}
\end{figure}

\begin{lemma} \label{lemma:bound_on_lambda_RS}
Suppose $|S| \ge \gamma n$. Then
\[
-\lambda(R_S) \ge \frac{\gamma}{\trel}.
\]
\end{lemma}

\begin{proof}
  Let $U_1, \ldots, U_k$ be the connected components of $V \setminus S$.
  By rearranging the rows and columns, we can observe that $Q[V \setminus S]$ is a block diagonal matrix with the blocks corresponding to the submatrices $Q[U_j]$ for $1 \le j \le k$ and the diagonal values $\lambda(Q[V\setminus S])$; see \Cref{fig:rs-construction}.
  Since the spectrum of a block diagonal matrix is the union of the spectra of its diagonal blocks, we get that
\begin{align*}
\spec(R_S) &= \{ \lambda(Q[V \setminus S]) \} \cup \spec(Q[V \setminus S]) \\
           &= \spec(Q[V \setminus S]) \\
           &= \bigcup_{j=1}^k \spec(Q[U_j]).
\end{align*}
By definition, this implies that
\begin{equation} \label{eq:lambda_RS}
	\lambda(R_S) = \max_{j \in \{1,\ldots,k\}} \lambda(Q[U_j]).
\end{equation}
Recall that for every $j \in \{1,\ldots,k\}$, $U_j$ induces a connected subgraph of~$G$; therefore, by \Cref{lemma:quasistationary-mean-hitting-time} with $S=V\setminus U_j$, 
\begin{equation*}
-\lambda(Q[U_j]) \ge \frac{\pi(V \setminus U_j)}{\trel} \ge \frac{\pi(S)}{\trel} \ge \frac{\gamma}{\trel},
\end{equation*}
since the stationary distribution $\pi$ of the walk $Y'$ is uniform and $|S| \ge \gamma n$.
By \Cref{eq:lambda_RS}, this implies that $-\lambda(R_S) \ge \gamma/\trel$, and the claim follows.
\end{proof}

Let~$\textinctioncont$ be the extinction time of the \emph{continuous-time} annihilation dynamics with the interaction rates given by $Q$.
The proof of the next lemma follows by simply retracing the analysis of Draief and Vojnovi\'{c}~\cite[Theorem 4.2]{draief2012convergence}, checking that their arguments still hold with our new matrix $R_S$, and applying \Cref{lemma:bound_on_lambda_RS}. For the sake of completeness, we give the proof in \Cref{apx:dv-proof}.

\begin{restatable}{lemma}{dvlemma}
\label{lemma:draief_vojnovic}
For any $\kappa>0$ and any initial configuration of the two-species annihilation dynamics with bias $\gamma > 0$, the extinction time satisfies
\begin{equation*}
	\Pr\left[\textinctioncont > \frac{(\kappa+1) \trel \ln n}{\gamma} \right] \le \frac{1}{n^\kappa}.
\end{equation*}
\end{restatable}

We now have a bound for the continuous-time process. To convert this to a bound for the discrete-time process, we can invoke folklore concentration bounds for Poisson random variables; see, e.g., the write-up of Canonne~\cite{cannonne2019poisson} for a proof of the next lemma.

\begin{lemma}\label{lemma:poisson-concentration}
Let $N \sim \operatorname{Poisson}(\mu)$. Then for any $x>0$
\[
\Pr[N \ge \Exp[N] + x ] \le \exp\left( -\frac{x^2}{x+\Exp[N] }\right ).
\]
\end{lemma}

We are now ready to prove \Cref{thm:extinction-time}.

\begin{proof}[Proof of \Cref{thm:extinction-time}]
For the continuous time process, the expected number of interactions in one unit of time is $1/2$, as $q_{u,v} = 1/(2m)$ for all $\{u,v\} \in E$ and the sum of independent Poisson random variables with rates $\mu_1$ and $\mu_2$ is a Poisson random variable with rate $\mu_1 + \mu_2$.
Thus, the number $N$ of interactions in the time interval of length
\[
\tau(\kappa) = \frac{(\kappa+1) \trel \ln n}{\gamma}
\]
is a Poisson random variable with mean $\tau(\kappa)/2$.
Thus by \Cref{lemma:poisson-concentration} with $x=\Exp[N]$, we get $\Pr[N>\tau(\kappa)]\leq \exp\left(-\frac{\tau(\kappa)}{2}\right)$.
Now by partitioning over the event $\{\textinctioncont > \tau(\kappa)\}$ and its complement and by \Cref{lemma:draief_vojnovic}, we have that
\begin{align}
  \Pr[\textinction > \tau(\kappa)] &\le \Pr[\textinctioncont > \tau(\kappa)] + \Pr[N > \tau(\kappa)] \\
  &\le \frac{1}{n^\kappa} + \exp\left(-\frac{\tau(\kappa)}{2}\right) \\
  &\le  \frac{1}{n^\kappa} + \left(\frac{1}{n}\right)^{(\kappa+1)\trel/\gamma} \le \frac{2}{n^\kappa},
\end{align}
as $\gamma\leq 1$ for any initial configuration and $\trel \ge 1$ by \Cref{lemma:spectral-gap-sandwich}.
\end{proof}

\subsection{Clearing time of the annihilation dynamics} \label{sec:clearing}

For $0 < \varepsilon < 1$, we define the \emph{$\varepsilon$-clearing time} of the two-species annihilation dynamics as the random time~$\tclear(X, \varepsilon)$ in which either the minority species is no longer present, or at most an $\varepsilon$-fraction of the nodes are non-empty. Formally,
\[
\tclear(X, \varepsilon) = \min \{ \textinction(X) \} \cup \{ t : |\calC_t| \ge (1-\varepsilon)n \}.
\]
In this section, we derive a tail bound for $\tclear(X,\varepsilon)$.

\begin{theorem}[Clearing time] \label{lemma:clearing}
For every~$\kappa \ge 1$, we have
\[
\Pr\left[\tclear(X,\varepsilon) > \frac{8(\kappa+1)  \trel \ln n}{\varepsilon} \right] \le \frac{2}{n^\kappa}.
\]
\end{theorem}

Importantly, the bound in \Cref{lemma:clearing} does not depend on the bias $\gamma>0$, but only on $\varepsilon$, which controls how ``empty'' of a configuration we are looking for.
This is in contrast with our previous bound on the extinction time (\Cref{thm:extinction-time}), which becomes quite large when~$\gamma$ is small (say~$1/n$).

\paragraph{Analysis.} For the sake of proving \Cref{lemma:clearing}, we replace the chain~$(X_t)_{t \ge 0}$ by an equivalent chain~$(Z_t)_{t \ge 0}$ taking values in~$\mathbb{R}$. Formally, for~$u \in V$, let
\begin{equation*}
	Z_t(u) := \begin{cases}
		+1 & \text{if } X_t(u) = \atype \\
		0 & \text{if } X_t(u) = \ctype \\
		-1 & \text{if } X_t(u) = \btype.
	\end{cases}
\end{equation*}

By definition of the dynamics, the quantity~$Z_t$ satisfies the following useful properties.
\begin{lemma}
If~$e_t = (u,v)$ is the interacting pair at time step~$t$, then
\begin{equation} \label{eq:Z_conservation}
	Z_t(u) + Z_t(v) = Z_{t+1}(u) + Z_{t+1}(v),
\end{equation}
\begin{equation} \label{eq:Z_contraction}
	|Z_{t+1}(u) - Z_{t+1}(v)| \leq 1,
\end{equation}
\begin{equation} \label{eq:Z_swapping}
	|Z_{t+1}(u) - Z_{t+1}(v)| = 1 \implies (Z_t(u), Z_t(v)) = (Z_{t+1}(v), Z_{t+1}(u)).
\end{equation}
\end{lemma}

Now, we consider a natural coupling~$(Z_t, Z_t')_{t \ge 0}$ defined from two arbitrary initial configurations $Z_0$ and $Z_0'$, and updated using the same schedule $\sigma = (e_t)_{t \ge 0}$ of interacting pairs. In other words, during the $t$-th interaction step with $e_t = (u,v)$, nodes $u$ and $v$ interact in both $Z_t$ and in $Z_t'$ (although these nodes may be in different states in the different chains and thus apply a different rule.)
In what follows, we say that~$Z_t$ {\em dominates}~$Z_t'$, and write~$Z_t \ge Z_t'$, if for every~$u \in V$, $Z_t(u) \geq Z_t'(u)$.

\begin{lemma} \label{lemma:domination}
	If~$Z_t \ge Z_t'$, then~$Z_{t+1} \ge Z_{t+1}'$.
\end{lemma} 
\begin{proof}
Let~$(Z_t,Z_t')$ such that~$Z_t \ge Z_t'$ and let~$e_t = (u,v)$ be the interacting pair in step~$t$. Assume, for the sake of contradiction, that $Z_{t+1}(u) \le Z_{t+1}'(u) - 1$. Note that by \Cref{eq:Z_conservation} and since~$Z_t \ge Z_t'$,
\begin{equation*}
	Z_{t+1}(u) + Z_{t+1}(v) = Z_t(u) + Z_t(v) \geq Z_t'(u) + Z_t'(v) = Z_{t+1}'(u) + Z_{t+1}'(v),
\end{equation*}
and therefore we must have~$Z_{t+1}(v) \ge Z_{t+1}'(v) + 1$. We distinguish between 3 cases:
\begin{itemize}
	\item If~$Z_{t+1}(u) = Z_{t+1}(v)$, then~$Z_{t+1}'(u) \ge Z_{t+1}(u) + 1 = Z_{t+1}(v) + 1 \ge Z_{t+1}'(v) + 2$, which yields a contradiction with \Cref{eq:Z_contraction}.
	\item If~$Z_{t+1}'(u) = Z_{t+1}'(v)$, then~$Z_{t+1}(u) \le Z_{t+1}'(u) - 1 = Z_{t+1}'(v) - 1 \le Z_{t+1}(v) - 2$, which yields a contradiction with \Cref{eq:Z_contraction}.
	\item By \Cref{eq:Z_contraction}, the only remaining case is that
	\begin{equation*}
		|Z_{t+1}(u) - Z_{t+1}(v)| = |Z_{t+1}'(u) - Z_{t+1}'(v)| = 1.
	\end{equation*}
	Therefore, by \Cref{eq:Z_swapping}, this implies
	\begin{equation*}
		Z_t(v) = Z_{t+1}(u) \le Z_{t+1}'(u)-1 = Z_t'(v)-1,
	\end{equation*}
	which yields a contradiction with the fact that~$Z_t \ge Z_t'$.
\end{itemize}
Therefore, it must be the case that~$Z_{t+1}(u) \ge Z_{t+1}'(u)$, and the same holds symmetrically for~$v$. This implies that~$Z_{t+1} \ge Z_{t+1}'$ which concludes the proof of \Cref{lemma:domination}.
\end{proof}

We can now prove \Cref{lemma:clearing}.

\begin{proof}[Proof of \Cref{lemma:clearing}]
We define the sets $\calA_t, \calB_t, \calC_t$ as before for the chain~$(Z_t)_{t \ge 0}$, and analogously, we write $\calA_t', \calB_t', \calC_t'$ for the chain~$(Z_t')_{t \ge 0}$.
Let $\alpha = |\calA_0|/n$ and $\beta = |\calB_0|/n$.
Note that if $\alpha+\beta \le \varepsilon$, then $|\calC_0| \ge (1-\varepsilon)n$ and $\tclear(X_0, \varepsilon) = 0$, and the statement  holds trivially.
Moreover, if $\alpha-\beta = \gamma \ge \varepsilon/4$, then $\tclear(X_0) \le \textinction(X_0)$, and the statement holds as a consequence of \Cref{thm:extinction-time}.
In what follows, we therefore restrict attention to the case that~$\alpha + \beta > \varepsilon$ and $\alpha-\beta < \varepsilon/2$.
Let
\begin{equation*}
	\beta' = \frac{\alpha+\beta}{2} - \frac{\varepsilon}{4}.
\end{equation*}
Since $\alpha+\beta > \varepsilon$, we have that $\beta' > 0$.
Moreover, since~$\alpha-\beta < \frac{\varepsilon}{2}$, we have that $\beta' < \beta$.
Let~$\calA_0' = \calA_0$, and let $\calB_0' \subseteq \calB_0$ with $|\calB_0'| = \lceil \beta' n \rceil$.
Consider the coupling $(Z_t,Z_t')_{t \ge 0}$ on initial configurations~$(\calA_0,\calB_0)$ and~$(\calA_0',\calB_0')$ respectively; clearly, $Z_0'$ dominates $Z_0$, i.e., $Z_0' \ge Z_0$.
By using \Cref{lemma:domination} inductively, we obtain that for every~$t\ge 0$, $Z_t' \ge Z_t$, and by definition, this implies
\begin{equation} \label{eq:domination}
	\calA_t \subseteq \calA_t'.
\end{equation}
Let $T = \textinction(Z')$. At time $T$, there have been exactly $|\calB_0'|$ annihilation events in the chain $(Z_t')_{0 \le t \le T}$, and hence
\begin{equation} \label{eq:aux}
	|\calA_T'| = |\calA_0'| - |\calB_0'| = \alpha n - |\calB_0'| \le (\alpha - \beta')n.
\end{equation}
Finally,
\begin{align*}
	|\calA_T| &\le |\calA_T'| & \text{(by \Cref{eq:domination})} \\
	&\le (\alpha - \beta')n & \text{(by \Cref{eq:aux})} \\
    &\le n \pa{\frac{\alpha-\beta}{2} + \frac{\varepsilon}{4}} & \text{(by definition of~$\beta'$)} \\
    &< \frac{\varepsilon n}{2}. &\text{(since $\alpha-\beta < \varepsilon/2$)}.
  \end{align*}
Since $|\calB_T| \le |\calA_T| \le n\varepsilon/2$, we have that $|C_T| \ge (1-\varepsilon)n$. Therefore, $\tclear(X, \varepsilon) \leq T$ and
\begin{equation} \label{eq:stochastic_dom}
	\Pr\left[\tclear(X,\varepsilon) > \frac{4}{\varepsilon} \cdot 2(\kappa+1) \trel \ln n \right] \le \Pr\left[T > \frac{4}{\varepsilon} \cdot 2(\kappa+1) \trel \ln n \right].
\end{equation}
Since for~$(Z_t')_{t \geq 0}$, the bias is~$\gamma' = \alpha - \beta' =  (\alpha-\beta)/2 + \varepsilon/4 \ge \varepsilon/4$,
we can apply \Cref{thm:extinction-time} to bound the right-hand side of \Cref{eq:stochastic_dom} and conclude the proof of \Cref{lemma:clearing}.
\end{proof}

\subsection{Stabilization time of the 4-state protocol}

We conclude this section by remarking that bounds on the extinction time can be used to bound the stabilization time of the 4-state protocol of B{\'e}n{\'e}zit et al.~\cite{benezit2009interval}.
Writing~$\{\wtype_0,\wtype_1,\stype_0,\stype_1\}$ to denote the 4 states, the protocol is given by the rules
\begin{align} \label{eq:4-state-majority}
\stype_i + \stype_{1-i} \to \wtype_{1-i} + \wtype_{i}, \qquad \stype_i + \wtype_{1-i} \to \wtype_{i} + \stype_{i},
\end{align}
and the additional ``swapping'' rules
\begin{equation}
	\stype_i + \wtype_{i} \to \wtype_{i} + \stype_{i}, \qquad \wtype_i + \stype_{i} \to \stype_{i} + \wtype_{i}, \qquad \wtype_{i} + \wtype_{1-i} \to \wtype_{1-i} + \wtype_{i},
\end{equation}
for $i \in \{0,1\}$. A node $v$ with input $f(v) = i$ initializes itself to state $\stype_i$.
A node in state $\stype_i$ or $\wtype_i$ outputs $i$.
Here, the bias at time~$t$ is defined between the ``strong'' opinions  $\stype_0$ and $\stype_1$:
\[
\gamma_t = \frac{|X_t(\stype_0) - X_0(\stype_1)|}{n}.
\]
Note that again, $\gamma_t = \gamma$ remains constant throughout the execution, and is the same as the bias of the inputs. Moreover, eventually, only one type of strong opinion remains (which will be the majority opinion).
In their analysis of the stabilization time of the 4-state protocol,
Draief and Vojnovi{\'c}~\cite{draief2012convergence} considered two phases of the protocol:
\begin{itemize}[noitemsep]
 \item The \emph{first phase} consists of the time interval when both strong opinions $\stype_0$ and $\stype_1$ are present.%
\item The \emph{second phase} consists of the time interval when only the majority strong opinion $\stype_i$ is present.%
\end{itemize}
The first phase is equivalent to the annihilation dynamics by projecting
\[
\stype_0 \mapsto \atype, \qquad \stype_1 \mapsto \btype, \qquad \wtype_i \to \ctype \textrm{ for } i \in \{0,1\}.
\]
The analysis of the second phase is almost identical to the analysis of the first phase~\cite{draief2012convergence}, and in fact, the
 second phase can be simply projected to the annihilation dynamics in a straightforward manner.
 Thus, the extinction time bound for the annihilation process yields the following bound for the stabilization time of the 4-state protocol.

\constantub*

\section{A space-efficient phase clock on graphs} \label{sec:clocks}

In this section, we describe a space-efficient phase clock. We first construct a protocol for producing {\em clock ticks} at regular intervals. We prove its correctness in \Cref{sec:clock_ticks_regular} under the assumption that every node is activated at the same rate, and then show how to relax this assumption in \Cref{sec:clock_ticks_general}.
Finally, we explain in \Cref{sec:global_phase_clock} how clock ticks can be used to create a global {\em phase clock}, which we will use to synchronize the cancellation-doubling protocol in \Cref{sec:main_proof}.

\subsection{Internal clock ticks on regular graphs} \label{sec:clock_ticks_regular}

We now assume that some  nodes carry \emph{clock tokens}.
We start by designing a protocol, which allows the clock tokens to generate (internal) clock ticks at a controlled frequency.
The construction is similar in spirit to the clocks used in the recent leader election protocol of Alistarh et al.~\cite{alistarh2025near} for general graphs.
However, for our protocol, the clock ticks are not generated by nodes, but rather clock tokens -- this change allows us to save space in the cancellation doubling protocol. Our rules are also somewhat simpler to analyze. %

\paragraph{Clock parameters.}
We first consider the case of regular graphs.
Fix a constant~$\kappa>1$ controlling the failure probability of the clock, and let~$\tickfreq \in \Omega(\trel \log n)$ be a parameter controlling the frequency of internal clock ticks.
Let $q = d/m = 2/n$ be the probability that a fixed note interacts in the next time step, and for~$x>0$, let $J(x)$ be the solution to the equation
\[
J(x) 2^{J(x)} = x.
\]
Note that for~$x>0$, this equation has a unique solution, and hence~$J(x)$ is well-defined; it is in fact the Lambert W function~\cite{corless1996lambert}, but with base 2.
We now define the main clock parameters
\[
H = H(\kappa) := \lceil \kappa \log n \rceil \qquad \textrm{ and } \qquad
K = K(\kappa) := \left \lceil J\left( \frac{q\tickfreq}{H} \right) \right \rceil.
\]
Observe that by the above choice of parameters, we have the inequalities
\begin{align}
q \tickfreq  \le HK2^K \le \Theta(q \tickfreq). \label{eq:hk-property} %
\end{align}

\paragraph{The internal clock protocol.}
We assume that the nodes can create a single random bit per interaction; this can be extracted from the scheduler by checking if a node is the initiator or the responder.

Recall that the clock tokens are circulated in the graph. Hence, we describe the protocol from the perspective of a single clock token $w$ rather than a fixed node.
A token $w$ can flip a coin that succeeds with probability $p = 2^{-K}$ by generating a single random bit for $K$ consecutive interactions, using $O(K)$ states. Note that the coin flips for a fixed token are independent (but coin flips of two different tokens may be correlated).
The rules for the clock token are now as follows:
\begin{itemize}
  \item The token makes consecutive $p$-coin flips and counts the number of successes modulo $H$.
  \item Once the token has made $H$ \emph{successful} $p$-coin flips in total (i.e., its success counter overflows), then the token generates a \emph{clock tick}.
\end{itemize}
This process is repeated throughout the lifetime of a clock token.
See \Cref{fig:internal_clock} for an illustration.

\begin{figure}
\centering
\includegraphics[width=0.5\textwidth]{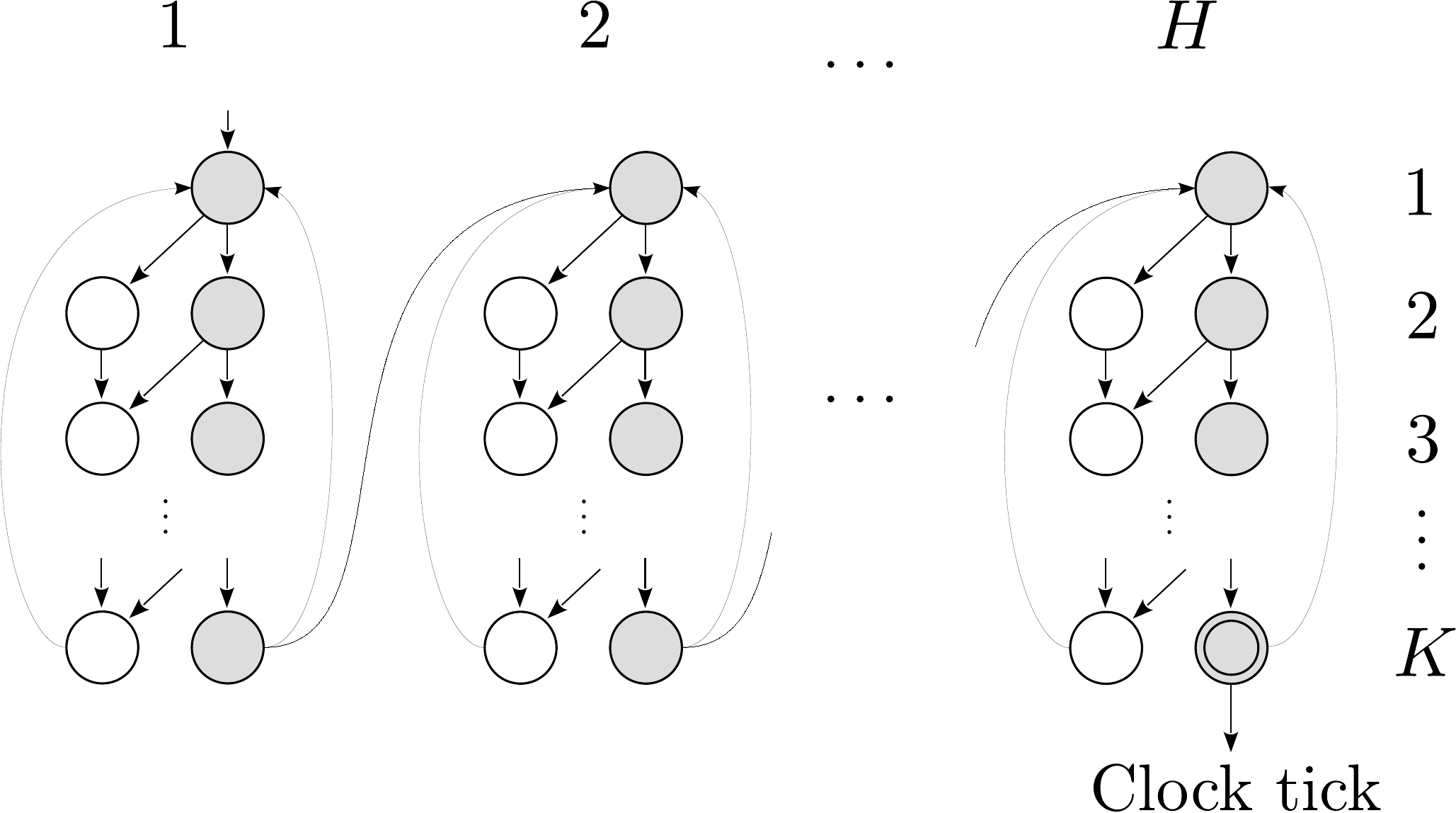}
\caption{An implementation of the internal clock using~$H(2K-1)$ states. Each column implements a $p$-coin with~$p=2^{-K}$. Tossing a $p$-coin takes exactly~$K$ local interactions. The coin flip is successful if the token remains on the right part of the column (depicted in gray) for all $K$ consecutive local interactions, which happens with probability~$p$. In that case, the token moves to the next column; otherwise, it goes back to the start of the \emph{same} column. Each transition requires one random bit. The token generates a clock tick after getting~$H$ (not necessarily consecutive) successful $p$-coin flips.}
\label{fig:internal_clock}
\end{figure}

\paragraph{Establishing frequency bounds.}
Let $(Y_i)_{i \ge 1}$ be a sequence of i.i.d.\ geometric random variables with success probability $p$.
By construction, the total number of $p$-coin flips a clock token needs in order to generate one clock tick has the same distribution as the random variable
\[
Y = \sum_{i=1}^H Y_i.
\]
Since performing a single $p$-coin flip takes exactly~$K$ interactions for the clock token, generating one clock tick requires~$K\cdot Y$ local interactions.
Recall that~$q = d/m = 2/n$ is the probability that a node holding the clock token $w$ is activated.
Writing $(X_i)_{i \ge 1}$ to denote a sequence of i.i.d.\ geometric random variables with success probability $q$, the number of global interactions it takes for the clock token to generate one clock tick is a random variable
\[
X = \sum_{i=1}^{K \cdot Y} X_i. %
\]
We now show that $X$ is concentrated around $\Theta(\tickfreq)$.

\begin{lemma}\label{lemma:internal-concentration}
  For all sufficiently large $n$,
\begin{enumerate}[label=(\alph*)]
\item the expectation of $X$ satisfies $\Exp[X] \in \Theta(\tickfreq)$, and
\item there exists a constant~$\lambda>1$ s.t. for~$n$ large enough $\P{X \in \br{\tfrac{1}{\lambda} \Exp[X], \lambda \, \Exp[X]} } \geq 1- 1/n^{\kappa}$.
\end{enumerate}
\end{lemma}
\begin{proof}
By linearity of expectation, $\Exp[Y] = H/p = H2^K$.
Claim (a) follows from Wald's equality, which gives
\begin{align}
\Exp[X] &= K \cdot \Exp[Y] \cdot \Exp[X_i] = HK2^K/q \label{eq:sum-of-xs},
\end{align}
so $\Exp[X] \in \Theta( \tickfreq )$ by \Cref{eq:hk-property}.
Let~$\lambda > 1$ be such that~$c(1/\sqrt{\lambda}) > 1$, where $c(x) := x - 1 - \ln x$ is defined in \Cref{lemma:sum-of-geometric} (e.g., $\lambda = 50$ is suitable). Furthermore, for any $y>0$, define the random variable
\begin{equation*}
	X^{(y)} := \sum_{i=1}^{K y} X_i.
\end{equation*}
Note that~$\Exp[X^{(y)}] = Ky \cdot \Exp[X_i]$, and thus by \Cref{lemma:sum-of-geometric}(c),
\begin{equation*}
	\P{X^{(y)} \in \br{\tfrac{1}{\sqrt{\lambda}} Ky \,\Exp[X_i], \sqrt{\lambda}  Ky\, \Exp[X_i]} } \geq 1- 2\exp\pa{ -Ky \cdot c(1/\sqrt{\lambda}) }.
\end{equation*}
Moreover, observe that for every~$y \in [\tfrac{1}{\sqrt{\lambda}}$ $\Exp[Y], \sqrt{\lambda} \, \Exp[Y]]$ we have that
\[
	\br{\tfrac{1}{\sqrt{\lambda}} Ky \, \Exp[X_i], \sqrt{\lambda} \, Ky \, \Exp[X_i]} \subseteq \br{\tfrac{1}{\lambda} \, K \Exp[Y] \Exp[X_i], \lambda \, K \Exp[Y] \Exp[X_i]} = \br{\tfrac{1}{\lambda} \Exp[X], \lambda \Exp[X]},
\]
and therefore
\begin{equation*} %
	\P{X^{(y)} \in \br{\tfrac{1}{\lambda} \Exp[X], \lambda \Exp[X]} } \geq 1-2 \exp\pa{ -\frac{\Exp[Y]}{\sqrt{\lambda}} \cdot K \cdot c(1/\sqrt{\lambda}) }.
\end{equation*}
Note that $\Exp[Y] = H2^K$ and as~$\trel\in\Omega(n)$,~$\lim_{n \rightarrow \infty} K = +\infty$, so for~$n$ large enough, $K2^K \geq \sqrt{\lambda}$, and we obtain
\begin{equation} \label{eq:xy_interval}
	\P{X^{(y)} \in \br{\tfrac{1}{\lambda} \Exp[X], \lambda \Exp[X]} } \geq 1-2 \exp\pa{ -H \cdot c(1/\sqrt{\lambda}) }.
\end{equation}
By the law of total probability, we have
\begin{equation*}
	\P{X \in [\tfrac{1}{\lambda} \Exp[X], \lambda \Exp[X]] } \geq \sum_{y = \left\lceil \tfrac{1}{\sqrt{\lambda}} \Exp[Y] \right\rceil }^{\left\lfloor \sqrt{\lambda}\Exp[Y] \right\rfloor } \P{ X^{(y)} \in [\tfrac{1}{\lambda} \Exp[X], \lambda \Exp[X]] } \cdot \P{Y=y}.
\end{equation*}
By \Cref{eq:xy_interval}, this implies
\begin{align*}
	\P{X \in [\tfrac{1}{\lambda} \Exp[X], \lambda \Exp[X]] } &\geq \pa{ 1-2 \exp\pa{ -H \cdot c(1/\sqrt{\lambda}) } } \sum_{y = \left\lceil \tfrac{1}{\sqrt{\lambda}} \Exp[Y] \right \rceil }^{ \left\lfloor \sqrt{\lambda}\Exp[Y] \right\rfloor}  \P{Y=y} \\
	&= \pa{ 1- 2\exp\pa{ -H \cdot c(1/\sqrt{\lambda}) } } \cdot \P{ Y \in \br{ \tfrac{1}{\sqrt{\lambda}} \Exp[Y],\sqrt{\lambda} \Exp[Y] } } \\
	&\geq \pa{ 1- 2\exp\pa{ -H \cdot c(1/\sqrt{\lambda}) } }^2 \geq 1- 4\exp\pa{ -H \cdot c(1/\sqrt{\lambda}) },
\end{align*}
where in the last line we have applied \Cref{lemma:sum-of-geometric}(c) to~$Y$.
Recall that by construction, $c(1/\sqrt{\lambda}) > 1$ and~$H\geq \kappa \log n$, so for~$n$ and~$\lambda$ large enough:
\begin{equation*}
	\P{X \in [\tfrac{1}{\lambda} \Exp[X], \lambda \Exp[X]] } \ge 1-\exp \pa{\kappa \log n} = 1-\frac{1}{n^\kappa},
\end{equation*}
which concludes the proof of \Cref{lemma:internal-concentration}.
\end{proof}

\subsection{Internal clock ticks: The case of general graphs}\label{sec:clock_ticks_general}

We now explain how to generalize the construction of the previous section to general graphs.
So far, we assumed that the probability $q(v)$ that a node $v$ containing the clock token is activated is the same for all nodes.
This is true for regular graphs under the scheduler that picks edges uniformly at random, but not for non-regular graphs.
Moreover, as the token traverses the graph, the probability that the token is sampled depends on its current location, which varies over time.

\paragraph{A scaling trick.}
If we are given some bounds $0 < q_0 \le q(v) \le q_1 < 1$ for the probability $q(v)$ that any node $v$ is activated in each step $t$, then we can generalize the above construction, with some loss in accuracy and space complexity.
In particular, if we have a graph with minimum degree at least $\delta$ and maximum degree at most $\Delta$, we already have the bounds
\[
q_0 = \frac{\delta}{m} \qquad \textrm{and} \qquad q_1 = \frac{\Delta}{m}.
\]
Suppose $T$ is the approximate tick frequency we want to achieve on an arbitrary (connected) graph $G$.
We now choose the scaling factor
\[
\theta := \frac{q_1}{q} \le \frac{q_1}{q_0} = \frac{\Delta}{\delta},
\]
where $q = 2/n$ as in the previous section.
We can now use the above protocol for regular graphs just by scaling the desired frequency by setting $\tickfreq \approx T \theta$.
This ensures that the time to generate a clock tick is  $\Omega(T)$ and at most $O(T\theta)$ time steps with high probability. The space complexity increases by a factor $O(\log \theta)$.
Therefore, the protocol will be efficient when the ratio of $\Delta/\delta$ is fairly small (i.e., in almost-regular graphs).

\paragraph{Internal clock for general graphs.}
We are now ready to state the lemma we use for our global phase clock.
Let $t_i(w)$ be the $i^{\text{th}}$ time a token $w$ generates its clock tick.

\begin{lemma}\label{lemma:nice-ticks}
  For every~$\kappa > 1$ and $\tickfreq \in \Omega(\trel \log n)$, there exists a value $\eta = \eta(\kappa) \in O(\Delta/\delta)$ such that~$\eta> 1$ and for any clock token $w$ and $i \ge 0$, we have
  \[
\Pr\left[t_{i+1}(w) - t_i(w) \notin [\tickfreq, \eta \tickfreq) \right] \le 1/n^{\kappa}.
  \]
  Moreover, the number of states used by a clock token is
  \[
O\left( \log n \cdot \log \left( \frac{\Delta}{\delta }\right) + \log n \cdot \log\left( \frac{ \tickfreq }{ n \log n} \right) \right).
  \]
\end{lemma}
\begin{proof}
Let $\lambda = \lambda(\kappa)$ be the constant from  \Cref{lemma:internal-concentration}.
We now set the tick frequency parameter for the protocol in \Cref{sec:clock_ticks_regular} to
\[
\tickfreq' :=  \lambda\theta\tickfreq
\]
and run the protocol in $G$ with this scaled parameter $\tickfreq'$.
Let $X$ be the number of global steps it takes for a fixed clock token $w$ to generate a local clock tick. This is again a sum of random variables
\[
X = \sum_{i=1}^{K\cdot Y} X_i,
\]
but now the summed random variables are not i.i.d., as the success probabilities $p_i$ of $X_i$ will depend on the location of the token when it is sampling its $i^{\text{th}}$ bit. However, we still have for any~$i \ge 0$,
\[
q_0 \le \Pr[X_i] \le q_1.
\]
For $j\in\{0,1\}$, let $X^{(j)}$ be the sum of $KY$ i.i.d.\ geometric random variables with success probability~$q_j$. We have
\[
X^{(1)} \preceq X \preceq X^{(0)}.
\]
That is, $X^{(0)}$ upper bounds the \emph{slowest rate} at which ticks happen and $X^{(1)}$ lower bounds the \emph{fastest rate}. Note that as in \Cref{eq:sum-of-xs}, we get
\[
\Exp\left[X^{(i)}\right] = K2^KH/q_i.
\]
Therefore, since $\tickfreq' = \tickfreq \theta \lambda$, we have that
\[
\Exp\left[X^{(0)}\right] = \Exp\left[X^{(1)}\right] \cdot \frac{q_0}{q_1} \qquad \textrm{and} \qquad \tickfreq'  \le \Exp\left[X^{(1)}\right] \in \Theta(\tickfreq).
\]
Applying \Cref{lemma:internal-concentration} to $X^{(0)}$ and $X^{(1)}$ gives
\begin{align*}
  \P{X \notin \left[ \frac{1}{\lambda}\Exp\left[X^{(1)}\right], \lambda \Exp\left[X^{(0)}\right] \right]} \le 2/n^\kappa
\end{align*}
for our choice of $\lambda$.
Thus, $X$ is at least $\tickfreq$ and $O( \theta \tickfreq )$ with high probability. In particular, this gives the bound for our lemma, as
  \[
\Pr\left[t_{i+1}(w) - t_i(w) \notin [\tickfreq, \eta \tickfreq) \right] \le 1/n^{\kappa}
  \]
  with $\eta = \lambda \theta$.
The bound on the space complexity follows by observing that $J(x) \le \log x$, because $J^{-1}(x) = x2^x \ge 2^x$, and therefore, the parameter $K$ is bounded by
\begin{align*}
  K &= \left\lceil J\left( \frac{q \tickfreq'}{H} \right) \right\rceil \\
    &\le \log \left(  \frac{q\lambda\theta \tickfreq}{H}  \right ) + 1 \\
    &=\log \left(  \theta \right) + \log \left( \frac {2\lambda\tickfreq}{nH}  \right ) + 1.
\end{align*}
The tokens use $O(HK)$ states.
Since $H \in \Theta(\log n)$, we get that the space complexity  is
\[
O\left( \log n \cdot \log \left( \frac{\Delta}{\delta }\right ) + \log n \cdot \log\left( \frac{ \tickfreq }{ n \log n} \right) \right). \qedhere %
\]
\end{proof}

\subsection{Globally synchronized phase clocks} \label{sec:global_phase_clock}

We are now ready to leverage the internal clocks defined in \Cref{sec:clock_ticks_regular,sec:clock_ticks_general} to construct a globally synchronized phase clock.
Here, it will be more convenient to talk about \emph{tokens} as the interacting agents, rather than the nodes (which host the agents).

Every token (clock token and non-clock token alike) in the population maintains a {\em phase} modulo~$\Phi$, for a given integer $\Phi>2$.
Our construction guarantees that at any given time step, the difference between any two phase values is at most one with high probability. Moreover, with high probability, every node spends sufficiently long time in the same phase.
This will be explained in detail later in \Cref{def:global_phase_clock}.

\paragraph{The phase clock protocol.}
Let~$\phi_t(v) \in \{0,1, \ldots, \Phi-1 \}$ be the phase of token~$v$ in time~$t$.
We run the internal clock tick protocol from above over some initially nonempty set $W$ of clock tokens.
Initially, all clock tokens in~$W$ are {\em active}.
When a token~$v$ interacts with some other token~$u$ at time step $t+1 >0$, it applies the following rules sequentially:
\begin{enumerate}
  \item If $v$ is an active clock token that generated a clock tick in this interaction, then $v$ increments its phase by one by setting $\phi_{t+1}(v) \gets \phi_t(v) + 1 \bmod \Phi$.
  \item If $\phi_t(u) = \phi_t(v) + 1 \bmod \Phi$, then $\phi_{t+1}(v) \gets \phi_t(u)$. If $v$ was an active clock token, it becomes \emph{inactive}.
\end{enumerate}
The idea is the following: consider a configuration in time~$t$ in which all tokens have the same phase~$\phi$. By Rule (2), the first active clock token to produce a clock tick after time~$t$ starts a broadcast, where the phase of every other token is incremented. As we will see, this broadcast completes before the same clock token generates a new clock tick with high probability.

Whenever the broadcast reaches an active clock token that have not produced a clock tick since time~$t$, this clock token becomes inactive, and stops producing clock ticks. This implies that only the ``fastest'' clock tokens remain in play. More formally, we will show that phase values have the following properties.

\begin{definition}[Global phase clock.] \label{def:global_phase_clock}
For integers~$R,\eta,\kappa > 0$, we say that $(\phi_t)_{t \ge 0}$ gives a~$(R,\eta,\kappa)$-clock if there exists a sequence of stopping times~$(r_i)_{i \geq 0}$ such that for any $k \ge 0$, with probability at least $1-k/n^\kappa$, each $0 \le i < k$ and $v \in V$ satisfy
\begin{enumerate}[label=(\alph*),noitemsep]
  \item monotonicity: $\phi_{t+1}(v) \neq \phi_t(v) \implies \phi_{t+1}(v) = \phi_t(v) + 1 \bmod \Phi$ for $t \in [r_i, r_{i+1})$,
  \item bounded delay: $R \le r_{i+1} - r_i \le 2 \eta R$,
  \item synchronization: $\phi_t(v) = i \bmod \Phi$ for $t \in [r_i, r_i+R]$, and
  \item agreement: $\phi_t(v) \in \{i, i+1 \} \mod \Phi$ for $t \in  [r_i, r_{i+1})$.
\end{enumerate}
\end{definition}

The monotonicity (a) ensures that clock tokens can only advance their phase one unit at a time.
The bounded delay property (b) ensures control on the frequency of
the global phase clock. %
The synchronization property (c) guarantees that all tokens reach a consensus on every new phase $i \mod \Phi$, which lasts for at least $R$ time steps.
Finally, the agreement property (d) guarantees that
the phase values are never far apart during the execution.
Figure~\ref{fig:phase-clock}a illustrates these properties.

\paragraph{Phase clock parameters.}
Fix a constant $\kappa>1$ controlling the failure probability.
Recall that~$\tbroadcast(v)$ denotes the broadcast time from source $v$, as defined in \Cref{sec:preliminaries}.
Let~$R = R(\kappa)$ be sufficiently large such that any $v \in V$ satisfies
\begin{equation} \label{eq:R_def}
  \Pr[ \tbroadcast(v) > R ] \le 1/n^{\kappa+1}.
\end{equation}
From \Cref{lemma:broadcast-faster-than-relaxing}
it follows that  $R \in \Omega(\trel \log n)$ suffices for this.
We set the internal clock parameters in \Cref{lemma:nice-ticks} such that $\tickfreq := 2R$, and there is~$\eta = \eta(\kappa) = O(\Delta/\delta)$ satisfying
\begin{equation} \label{eq:nice-ticks}
  \Pr\left[t_{i+1}(w) - t_i(w) \notin [\tickfreq, \eta \tickfreq) \right] \le \frac{1}{n^{\kappa+2}},
\end{equation}
where $t_i(x)$ denotes the time at which the clock token $x$ creates its $i^{\text{th}}$ clock tick. The existence of such parameters is guaranteed by \Cref{lemma:nice-ticks}.

\paragraph{Synchronization steps.}
We say that $r > 0$ is a \emph{synchronization step} if all nodes agree on the phase at time step $r$ (that is, $\phi_r(\cdot)$ is a constant function) and $\phi_{r-1} \neq \phi_r$. The $i^{\text{th}}$ synchronization step $r_i$ is defined by $r_0 := 0$ and
\[
r_{i+1} := \min \{ r > r_i : r \textrm{ is an synchronization step} \}
\]
for $i \ge 0$.
Our goal is to show that~$(\phi_t)_{t\ge 0}$ gives a global phase clock in the sense of \Cref{def:global_phase_clock} with respect to the stopping times $(r_i)_{i \ge 0}$.
Since clock ticks interact with the global phase clock only as long as the corresponding clock token is active, we define (for the sake of our analysis)
\[
t'_k(x) := \begin{cases}
t_k(x) & \textrm{if the token $x$ is an active clock token at time $t_k(x)$,} \\
+\infty & \textrm{otherwise.}
\end{cases}
\]
Finally, we define $t_0^\star := 0$ and
\[
t_k^\star := \min \{ t'_k(w)  : w \in W\}.
\]
In words, the time $t_k^\star$ is the first time step in which any \emph{active} clock token generates its $k^{\text{th}}$ tick.
Clearly, $t_k^\star \ge \min \{ t_k(w)  \} $.

\begin{figure}[t]
  \begin{center}
\includegraphics{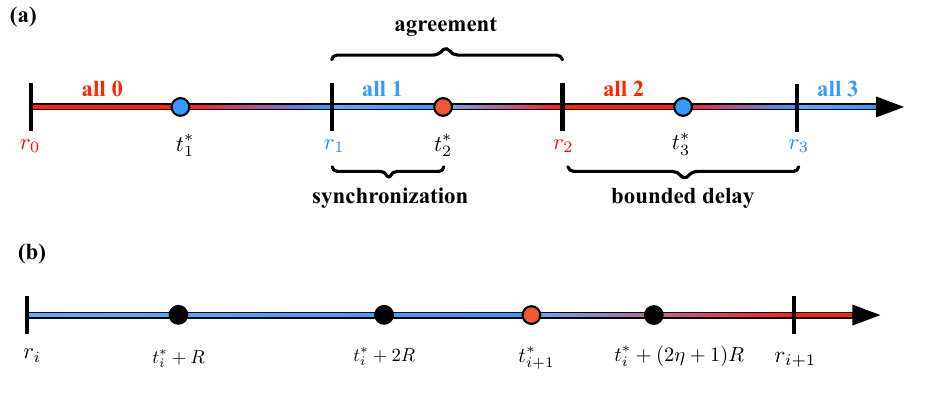}
\end{center}
\caption{Illustration of a phase clock execution with $\Phi=4$. The time $t^*_i$ indicates the first time when some clock token makes its $i^{\text{th}}$  internal tick and $r_i$ denotes the $i^{\text{th}}$ synchronization step. (a)~The synchronization, agreement and bounded delay properties. Agreement property ensures that  during an interval $[r_i, r_{i+1})$ all values are within distance one from one another (w.h.p.). Bounded delay property ensures that the consecutive synchronization steps are (w.h.p.) at controlled distance  $R \le r_{i+1}-{r_i} \le \eta R$.
(b)~The time steps used in the analysis for a single interval $[r_i, r_{i+1})$.}\label{fig:phase-clock}
\end{figure}

\paragraph{Proof strategy.}
As the clock construction may probabilistically fail, the proof is somewhat intricate.
Nevertheless, our intuitive proof strategy is fairly straightforward.
We want to show that (w.h.p.) that the values $t_i^\star$ and $r_i$ interleave nicely as follows:
\[
t_i^\star \le r_i < t_{i+1}^\star
\]
with controlled gaps between these time steps.
This property intuitively holds assuming the correct operation of the clock; Figure~\ref{fig:phase-clock}a illustrates how these values will relate (with high probability).

\paragraph{The analysis.}
For the analysis, we now consider a fairly technical looking definition that helps us formalize the behaviour of the clock.
Let $W_i \subseteq W$ be the set of \emph{active} clock tokens at time $r_i$.
We say that a phase $i \ge 0$ is \emph{good} if all the following events occur:
\begin{enumerate}[label=(\roman*),noitemsep]
\item $t_{i+1}^\star \ge t_i^\star + 2R$,
\item $r_i < t_i^\star + R$,
\item $t'_{i}(w) \le r_{i} < t'_{i+1}(w)$ for $w \in W_i$, i.e.,  clock tokens active at time step $r_i$ have ticked exactly $i$ times,
\item $t_{i+1}^\star < t_i^\star + (2\eta+1)R$,
\item $t_{i+1}^\star \le r_{i+1}$, and
\item $t_i^\star, r_i, t_{i+1}^\star$ are finite.
\end{enumerate}
Let $\mathcal{E}_i$ denote the event that phase $i$ is good.
We write $\mathcal{E}_i^\text{(i)}, \ldots, \mathcal{E}_i^\text{(vi)}$ for the above subevents of $\mathcal{E}_i$.
The event implies that $t_i^\star$, $r_i$ and $t_{i+1}^\star$ interleave nicely. Figure~\ref{fig:phase-clock}b illustrates the intended relationship between these time steps.
We now bound the conditional probability that the event $\mathcal{E}_i$ happens.

\begin{lemma}\label{lemma:clock-inductive-step}
For any $i > 0$, we have
\[
\Pr[ \overline{\mathcal{E}_{i}} \mid \mathcal{E}_{i-1} ] \le \frac{2}{n^{\kappa+1}}.
\]
\end{lemma}
\begin{proof}
If $\mathcal{E}_{i-1}$ holds, then by subevent (vi) of $\mathcal{E}_{i-1}^\text{(vi)}$, the time $t_i^\star$ is finite. Therefore, there exists a clock token $x$ such that $t_i(x) = t_i^\star$.
Let~$t^{\mathrm{broadcast}}$ be the time in which a broadcast started at~$x$ in time~$t_i^\star$ completes. That is, $t^{\mathrm{broadcast}} = t_i^\star + \tbroadcast(x)$.
Let~$A_1, A_2$ be the two following events:
\begin{equation*}
  A_1 := \left\{ t^{\mathrm{broadcast}} \leq t_i^\star + R \right\},
\end{equation*}
and
\begin{equation*}
  A_2 := \bigcap_{w \in W} \left\{ t_{i+1}(w) \in [\tickfreq, \eta \tickfreq ) \right\},
\end{equation*}
By \Cref{eq:nice-ticks} and a union bound, we have
\begin{equation} \label{eq:A_2-bound}
  \Pr\Big[\overline{A_2}\Big] \leq n \cdot \frac{1}{n^{\kappa+2}} = \frac{1}{n^{\kappa+1}}.
\end{equation}
Note that by definition of~$R$ in \Cref{eq:R_def}, we also have $\Pr\Big[\overline{A_1}\Big] \leq n^{-(\kappa+1)}$. We will show that
\begin{equation*}
  \mathcal{E}_{i-1} \cap A_1 \cap A_2 \implies \mathcal{E}_i,
\end{equation*}
and the claim will follow by a union bound over~$A_1$ and~$A_2$.
We now check each subevent (i)--(vi):
\begin{enumerate}[label=(\roman*)]

\item Recall that the subevent $\mathcal{E}_i^\text{(i)}$ is $t_{i+1}^\star \ge t_i^\star + 2R$.
Since~$\tickfreq = 2R$, we have
\begin{align*}
\left\{ t_{i+1}^\star < t_{i}^\star + 2R \right\} &\implies \bigcup_{w \in W} \left\{ t'_{i+1}(w) < t_i^\star + \tickfreq \right\} & \text{(by definition of $t_i^\star$)} \\
  &\implies \bigcup_{w \in W} \left\{ t_{i+1}(w) < t_i^\star + \tickfreq \right\} & \text{(since $t_i(w) \leq t_i'(w)$)} \\
  &\implies \bigcup_{w \in W} \left\{ t_{i+1}(w) < t_i(w) + \tickfreq \right\} & \text{(since $t_i^\star \leq t_i(w)$)} \\
  &\implies \overline{A_2},
\end{align*}
or equivalently, $A_2 \implies \mathcal{E}_i^\text{(i)}$.

\item Recall that  $\mathcal{E}_i^\text{(ii)}$ is  $r_i < t_i^\star + R$.
Assume that  $\mathcal{E}_i^{(i)}$ holds.
By $\mathcal{E}_{i-1}^\text{(iii)}$, all active clock tokens at time $r_{i-1}$ have made exactly $i-1$ clock ticks. This implies that in the interval $(r_{i-1},t_i^\star)$, none of the active clock tokens produces a clock tick. Since~$\phi_{r_i}(\cdot)$ is constant by definition, this also implies (by Rules (1) and (2)) that~$\phi_t(\cdot)$ remains constant during this interval.
In time~$t_i^\star$, the token $x$ produces a clock tick and increments its phase, and the new phase value is broadcast by Rule (2). By construction, this broadcast completes in time~$t^{\mathrm{broadcast}}$.

We have that $t^{\mathrm{broadcast}} < t_i^\star+R$ by~$A_1$, and $t_i^\star+R < t_{i+1}^\star$ by $\mathcal{E}_i^\text{(i)}$. Hence, all active clock tokens can produce at most one tick in the interval~$(t_i^\star,t^{\mathrm{broadcast}})$. In that case, all tokens have the same phase in round~$t^{\mathrm{broadcast}}$, and thus~$r_i \leq t^{\mathrm{broadcast}} < t_i^\star+R$. Thus, we get
\begin{equation*}
  \mathcal{E}_{i-1} \cap A_1 \cap A_2 \cap \mathcal{E}_i^\text{(i)} \implies \mathcal{E}_i^\text{(ii)}.
\end{equation*}

\item Recall that  $\mathcal{E}_i^\text{(iii)}$ says that clock tokens active at time step $r_i$ have ticked exactly $i$ times.
Assume that $\mathcal{E}_i^\text{(i)}$ and $\mathcal{E}_i^\text{(ii)}$ hold.
This implies that~$r_i \in [t_i^\star, t_{i+1}^\star)$, and by definition, all active clock tokens produce at most one clock tick in this interval.
Moreover, this implies that the phase of every token is incremented exactly once between times~$t_i^\star$ and~$r_i$.

Let~$w$ be any clock token, and consider the time~$s \in [t_i^\star,r_i]$ in which the phase of~$w$ is incremented. If~$w$ produces a clock tick in time~$s$, then $w$ produces exactly one clock tick in the interval~$[t_i^\star, t_{i+1}^\star)$. Otherwise, by Rule (2), $w$ is deactivated in time~$s$ and is not active in time~$r_i$. Since all active clock tokens had produced exactly $i-1$ clock ticks in time~$r_{i-1}$ by $\mathcal{E}_{i-1}^\text{(iii)}$, we can conclude that those that are still active in time~$r_i$ have produced exactly~$i$ clock ticks, and therefore $\mathcal{E}_i^\text{(iii)}$ holds.

\item Recall that event $\mathcal{E}_i^\text{(iv)}$ is~$t_{i+1}^\star < t_i^\star + (2\eta+1)R$.
Assume that $\mathcal{E}_i^\text{(i)}$, $\mathcal{E}_i^\text{(ii)}$ and $\mathcal{E}_i^\text{(iii)}$ hold.
Then $r_i < \infty$ and for every~$w \in W_i$, $t'_i(w) \le r_i \le t_i^\star + R$. Since~$\tickfreq = 2R$, we get
\begin{align*}
\left\{ t_{i+1}^\star \ge t_i^\star + (2\eta+1)R \right\} &\implies \bigcap_{w \in W} \left\{ t_{i+1}'(w) \ge t_i^\star + R + \eta\tickfreq \right\} & \text{(by def. of $t_i^\star$)} \\
&\implies \bigcap_{w \in W_{i+1}} \left\{ t'_{i+1}(w) \ge t_i^\star + R + \eta\tickfreq \right\} & \\
&~= \bigcap_{w \in W_{i+1}} \left\{ t_{i+1}(w) \ge t_i^\star + R + \eta\tickfreq \right\} & \text{(by def. of~$t_{i+1}'(w)$)} \\
&\implies \bigcap_{w \in W_{i+1}} \left\{ t_{i+1}(w) \ge t_i(w) + \eta\tickfreq \right\} & \text{($t_i(w) \le t_i^\star + R$)} \\
&\implies \overline{A_2},
\end{align*}
and finally, we obtain that
\begin{equation*}
  \mathcal{E}_{i-1} \cap A_1 \cap A_2 \cap \mathcal{E}_i^\text{(i)} \cap \mathcal{E}_i^\text{(ii)} \cap \mathcal{E}_i^\text{(iii)} \implies \mathcal{E}_i^\text{(iv)}.
\end{equation*}

\item Recall that event $\mathcal{E}_i^\text{(v)}$ is $t_{i+1}^\star \le r_{i+1}$. Assume that in addition to $\mathcal{E}_{i-1}$, the events $\mathcal{E}_i^\textrm{(i)}$ to $\mathcal{E}_i^\textrm{(iv)}$ hold. Then we have by definition of $r_i$ that $\phi_t(\cdot) = i \mod \Phi$ for every $t \in [r_i, t_{i+1}^\star)$. This shows that $r_{i+1} \ge t_{i+1}^\star$.

\item The event $\mathcal{E}_i^\text{(vi)}$ is a direct consequence of the event $\mathcal{E}_i^\text{(i)} \cap \mathcal{E}_i^\text{(ii)} \cap \mathcal{E}_i^\text{(iii)} \cap \mathcal{E}_i^\text{(iv)}$.
\qedhere
\end{enumerate}
\end{proof}

\begin{lemma}\label{lemma:clock-induction}
For any $k \ge 0$, we have
\[
\Pr\left[ \bigcup_{i=0}^k \overline{\mathcal{E}_i} \right] \le \frac{k}{n^{\kappa}}.
\]
\end{lemma}
\begin{proof}
Clearly, $t_1 > t_0$, $r_0 = t_0 < R$, so events $\mathcal{E}_0^\text{(i)}$ to $\mathcal{E}_0^\text{(iii)}$ and $\mathcal{E}_0^\text{(vi)}$  occur with probability 1.
Events $\mathcal{E}_0^\textrm{(iv)}$ and  $\mathcal{E}_0^\textrm{(v)}$ hold with probability at least $1-1/n^{\kappa+1}$ each, by the same arguments as in the proof of \Cref{lemma:clock-inductive-step} with $i=0$; hence, overall
\begin{equation} \label{eq:base_case}
\Pr[\overline{\mathcal{E}_0}] \le \frac{2}{n^{\kappa+1}}.
\end{equation}
Finally, we have
\begin{align*}
  \Pr\left[ \bigcap_{i=0}^k \mathcal{E}_i \right] &= \prod_{i=0}^k \Pr\Bigg[ \mathcal{E}_i \mid \bigcap_{j<i} \mathcal{E}_j \Bigg] & \text{(by the chain rule)} \\
  &=  \Pr\left[ \mathcal{E}_0 \right] \cdot \prod_{i=1}^k \Pr\left[ \mathcal{E}_i \mid \mathcal{E}_{i-1} \right] & \text{(by Markov property)} \\
  &\geq \pa{1-\frac{2}{n^{\kappa+1}}}^{k+1} & \text{(by \Cref{lemma:clock-inductive-step} and \Cref{eq:base_case})} \\
  &\ge 1-\frac{k}{n^{\kappa}}. \qedhere
\end{align*}
\end{proof}

\begin{theorem}\label{thm:clock-thm}
For every~$\kappa>1$, there exists $R(\kappa) \in \Omega( \trel \log n )$ and~$\eta(\kappa) \in O(\Delta/\delta)$ such that $(\phi_t)_{t \ge 0}$ gives a $(R,\eta,\kappa)$-clock.
Moreover, every token uses $\Phi$ states to encode its phase value.
The number of states used by any clock token is
  \[
O\left( \log n \cdot \pa{ \log \Big( \frac{\Delta}{\delta }\Big) + \log\Big( \frac{ R }{ n \log n} \Big) } \right).
  \]
\end{theorem}
\begin{proof}
The space complexity bound follows from \Cref{lemma:nice-ticks} and the fact that storing $\phi(v)$ takes $\Phi \in O(1)$ states.
Let $k \ge 0$. By \Cref{lemma:clock-induction}, we have that events $\mathcal{E}_0, \ldots, \mathcal{E}_{k}$ happen with probability~$k/n^\kappa$. We now show that if these events happen, the properties (a)--(d) are satisfied for each $0 \le i < k$:
\begin{enumerate}[label=(\alph*)]
  \item The monotonicity property follows immediately for all $t \ge 0$ from Rule (1) and Rule (2).

  \item The bounded delay property follows from the events $\mathcal{E}_{i}$ and $\mathcal{E}_{i+1}$ as we have
  \begin{equation} \label{eq:bounded_delay1}
  r_{i+1} - r_i  \stackrel{\text{(v)}}{\ge} t_{i+1}^\star - r_i \stackrel{\text{(ii)}}{\ge} t_{i+1}^\star - t_i^\star - R \stackrel{\text{(i)}}{\ge} R,
  \end{equation}
  and
   \begin{equation*}
  r_{i+1} - r_i  \stackrel{\text{(ii)}}{\le} t_{i+1}^\star + R - r_i \stackrel{\text{(v)}}{\le} t_{i+1}^\star + R - t_i^\star \stackrel{\text{(iv)}}{\le} 2\eta R.
  \end{equation*}

\item For the synchronization property, we need to show that for $t \in [r_i, r_i + R]$, we have $\phi_t(\cdot) = i \bmod \Phi$.
  By~$\mathcal{E}_i^\text{(iii)}$, all active clock tokens at time step $r_i$ have ticked exactly $i$ times.
  Therefore, their phase is~$i \bmod \Phi$. By definition of $r_i$, $\phi_{r_i}(\cdot)$ is constant, and hence it is equal to~$i \bmod \Phi$ for every token.
  This remains the case until time~$t_i^\star$, and the claim follows since we have seen in \Cref{eq:bounded_delay1} that $t_{i+1}^\star - r_i \ge R$.

\item For the agreement property, we need to show that $\phi_t(v) \in \{i, i+1 \} \mod \Phi$ for $t \in  [r_i, r_{i+1})$.
  We have by $\mathcal{E}_i^\text{(iii)}$ that all active clock tokens at time step $r_i$ have ticked exactly $i$ times. Since
  \begin{align*}
    r_{i+1} \stackrel{\text{(ii)}}{\le} t_{i+1}^\star + R \stackrel{\text{(i)}}{<} t_{i+2}^\star,
  \end{align*}
  we get that no active clock token during the interval $[r_i, r_{i+1})$ ticks twice. Hence, by Rule (1) and Rule (2) no node increments its phase counter by more than one during the interval $[r_i, r_{i+1})$, which implies the claim.
\end{enumerate}
Hence, $(\phi_t)_{t \ge 0}$ gives a $(R, \eta, \kappa)$-clock, which concludes the proof of \Cref{thm:clock-thm}.
\end{proof}

\section{A fast and space-efficient protocol for exact majority} \label{sec:main_proof}

\newcommand{\algone}{\hyperref[sec:fast_protocol]{Algorithm~1}\xspace}
\newcommand{\algtwo}{\hyperlink{sec:failure_detection}{Algorithm~2}\xspace}
\newcommand{\abort}{\texttt{Abort}\xspace}
\newcommand{\awins}{$\atype$-\texttt{wins}\xspace}
\newcommand{\bwins}{$\btype$-\texttt{wins}\xspace}

In this section, we explain how to achieve the upper bound given in \Cref{thm:fast_upper_bound}.
In \Cref{sec:fast_protocol}, we describe a first protocol that removes the minority opinion from the system fast with high probability using the annihilation process from \Cref{sec:cancellation} and the phase clock from \Cref{sec:clocks}.
In \Cref{sec:failure_detection}, we extend this protocol to make sure that almost surely all outputs stabilize to the correct majority value by using the 4-state always-correct protocol as backup. This will give us our main result:

\fastubgeneral*

\subsection{A fast protocol that works with high probability: Algorithm 1} \label{sec:fast_protocol}

Using the phase clocks we have constructed, we implement a synchronized cancellation--doubling protocol to amplify the input bias fast, following the idea described in \Cref{ssec:overview}.
As before, it will be convenient to describe the protocol with tokens as the interacting agents, rather than the nodes, as the tokens traverse the graph by swapping places and update their states when interacting.

\paragraph{Token and protocol state variables.}
Every token $u$ has a {\em type} $\omega(u) \in \{ \atype, \btype, \smallatype, \smallbtype, \ctype, \bot \}$ and a {\em phase} $\phi(u) \in \{0,\ldots,\Phi-1\}$, where~$\Phi=4$.
Tokens with type~$\bot$ are called {\em clock tokens}, and all other token types are {\em opinion tokens}.

Clock tokens run the internal clock protocol of \Cref{sec:clock_ticks_regular}, and all tokens run the global phase clock protocol of \Cref{sec:global_phase_clock}.
In addition to their type and phase, each opinion token $u$ maintains an {\em iteration counter} $\vartheta(u) \in \{0,\ldots, \lceil 2 \log n \rceil \}$, which will be useful for error detection later on.

\paragraph{Interaction rules.}
When the protocol starts, every node holds an opinion token~$u$ with type~$\omega(u) \in \{\atype,\btype\}$.
Nodes with input 0 have the token $\atype$ and nodes with input 1 have the $\btype$ token.
Initially, all tokens are opinion tokens.
Every token $u$ initializes its phase $\phi_0(u) \gets 0$ and iteration counter~$\vartheta_0(u) \gets 0$ to be zero.
Each time an opinion token $u$ interacts at time $t+1$, the following rules are applied:
\begin{itemize}

  \item If the phase of $u$ is incremented (as a consequence of the protocol in \Cref{sec:global_phase_clock}),
then it also increments its iteration counter by one, i.e., $\vartheta(u) \gets \vartheta(u) + 1$.
This implies that for every opinion token~$u$, we have~$\vartheta_t(u) = \phi_t(u) \mod \Phi$ as long as~$\vartheta(u)$ does not overflow.

\item When an opinion token of type $\smallatype$ or $\smallbtype$ gets its phase incremented from an \emph{odd phase} to an \emph{even phase}, it transforms into a token of type $\atype$ or $\btype$, respectively.
\end{itemize}
After applying the above rules, two interacting \emph{opinion} tokens $u$ and $v$ further apply one of the following sets of rules, depending on their current phase $\phi$ and the value of their iteration counters $\vartheta$:

\begin{enumerate}
\item {\em Initialization.} In order to introduce some clock tokens into the network at the beginning of execution, two opinion tokens with~$\vartheta(u) = \vartheta(v) = 0$ apply the  rules
\begin{align*}
	\atype + \btype \to \bot + \ctype, & \qquad \btype + \atype \to \bot + \ctype.
\end{align*}
Any clock token~$w$ produced in this way is initialized with phase~$\phi(w) \gets 0$.

\item {\em Cancellation.} If $\phi(u) = \phi(v)>0$ is \emph{even}, then $u$ and $v$ apply the  rules:
\begin{align*}
	\atype + \btype \to \ctype + \ctype, & \qquad \btype + \atype \to \ctype + \ctype.
\end{align*}

\item {\em Doubling.} If $\phi(u) = \phi(v)$ is \emph{odd}, then $u$ and $v$ apply the  rules
\begin{align*}
	\atype + \ctype \to \smallatype + \smallatype, & \qquad \ctype + \atype \to \smallatype + \smallatype, \\
	\btype + \ctype \to \smallbtype + \smallbtype,  & \qquad \ctype + \btype \to \smallbtype + \smallbtype.
\end{align*}
\end{enumerate}
Finally, at each interaction, and after applying all the aforementioned rules, any two interacting tokens swap their places in the graph.

\paragraph{Analysis.}
We say that a token is a \emph{strong} opinion token if its type is $\atype$ or $\btype$, and a \emph{weak} opinion token if its type is $\smallatype$ or $\smallbtype$. Without loss of generality, we assume throughout that input 0 is the majority value, as both protocols will be symmetric with respect to the input values. Thus, we refer to $\atype$ and $\smallatype$ as majority opinion tokens, and to $\btype$ and $\smallbtype$ as minority opinion tokens.
First, we bound the number of clock tokens in the system at all time.

\begin{lemma} \label{lemma:not_too_many_clocks}
	For every~$t \geq 0$, the number of clock tokens (i.e., tokens with state~$\bot$) is at most~$n/2$.
\end{lemma}
\begin{proof}
We say a strong opinion token $u$ is {\em primary} if $\vartheta(u)=0$.
By definition, a clock token is produced only when two primary tokens interact. There are initially $n$ primary tokens in the population, and no new primary tokens can ever be created by the protocol. Moreover, when two primary tokens interact, they produce a single clock token and a single $\ctype$-token.  Therefore, the maximum amount of clock tokens that can be produced is~$n/2$.
\end{proof}

Fix a constant~$\kappa>1$ controlling the failure probability.
Let~$R = R(\kappa) \in \Theta(\trel \log n)$ and $\eta = \eta(\kappa) \in \Theta(\Delta/\delta)$ be sufficiently large so that \Cref{thm:clock-thm} gives a $(R,\eta,\kappa+1)$-clock, and such that
\begin{equation} \label{eq:R_def2}
	R \geq 80(\kappa+2)\trel \ln n.
\end{equation}
Now, we show that these definitions ensure that the system is correctly synchronized, as long as at least one clock token is produced by the ``Initialization'' rule.
\begin{lemma} \label{lem:initialization}
If all nodes have the same input, then \algone stabilizes in one step. Otherwise, $(\phi_t)_{t \ge 0}$ gives a $(R,\eta,\kappa+1)$-clock (in the sense of \Cref{def:global_phase_clock}) with high probability. %
\end{lemma}
\begin{proof}
First, consider the case that all nodes have the same input.
In this case, all nodes and tokens start in the same same state with $\phi(u) = \vartheta(u) = 0$.
From this configuration, phases and iteration counters cannot be incremented, as there are no clock tokens. Moreover, since iteration counters $\vartheta(\cdot)$ are equal to~$0$, the only rule that can be applied by opinion tokens is the ``Initialization'' rule; but since there are only tokens of one type (corresponding to the input value), this rule has no effect. Therefore, the system is stable from step~$1$, and by construction, all tokens have the correct output.

Next, consider the case that there are nodes with different inputs.
In that case, there is at least one edge with an $\atype$-token and a $\btype$-token, and once one such edge is sampled, there is a clock node by the ``Initialization'' rule.
The time $t_1$ to sample one such edge is stochastically dominated by the geometric random variable $X \sim \Geom(\zeta/m)$, where $\zeta>0$ is the edge expansion of $G$. By \Cref{lemma:sum-of-geometric} such an edge is sampled in $O(m/\zeta \log n)$ steps with high probability. By \Cref{lemma:spectral-gap-sandwich}, we get that $t_1 \in O(\trel \log n)$ with high probability, as well.

The claim follows using similar arguments as in the proof of \Cref{lemma:clock-inductive-step} given in \Cref{sec:global_phase_clock}.
By \cref{lemma:nice-ticks}, there is a time~$t_2<t_1+2 \eta R$ such that one of the clock tokens produces a clock tick, starting a broadcast. The broadcast completes with high probability before time~$t_3<t_2+R$ by \Cref{lemma:broadcast-faster-than-relaxing}, incrementing the phase of every token to~$1$. With high probability, no clock token with phase~$1$ produces a clock tick before time~$t_3$ by \Cref{lemma:nice-ticks}.
From time~$t_3$ onwards, the iteration counter of every opinion token is always greater than~$1$, and no additional clock tokens can be produced.
Therefore, we can apply \Cref{thm:clock-thm} and the claim follows for our choice of~$R$ and~$\eta$.
\end{proof}

From now on, we restrict attention to the case that nodes have different inputs, and therefore $(\phi_t)_{t \ge 0}$ gives a $(R,\eta,\kappa+1)$-clock by \Cref{lem:initialization}. Let~$(r_i)_{i \geq 0}$ be the sequence of stopping times mentioned in \Cref{def:global_phase_clock}.
For every~$i\geq 0$, we define an annihilation process~$Z^{(i)} = (Z_t^{(i)})_{t \geq r_i}$ starting in time~$r_i$ and coupled with our protocol. Specifically, at time~$r_i$, every token~$u$ from our protocol is mapped to a token in $\{\atype, \btype, \ctype\}$ in the corresponding annihilation process  as follows:
\begin{itemize}
	\item If~$i$ is even, then 
	\begin{equation} \label{eq:coupling_even}
		Z_{r_i}^{(i)}(u) := 
		\begin{cases}
			\atype & \text{if } \omega_{r_i}(u) = \atype, \\
			\btype & \text{if } \omega_{r_i}(u) = \btype, \\
			\ctype & \text{if } \omega_{r_i}(u) \in \{\smallatype,\smallbtype,\ctype,\bot\}.
		\end{cases}
	\end{equation}

	\item If~$i$ is odd, and the number of $\ctype$-tokens exceeds the number of tokens with states in~$\{\atype,\btype\}$ by at least~$n/10$, then
	\begin{equation} \label{eq:coupling_odd}
		Z_{r_i}^{(i)}(u) := 
		\begin{cases}
			\atype & \text{if } \omega_{r_i}(u) \in \{\atype,\btype\}, \\
			\btype & \text{if } \omega_{r_i}(u) = \ctype, \\
			\ctype & \text{if } \omega_{r_i}(u) \in \{\smallatype,\smallbtype,\bot\}.
		\end{cases}
	\end{equation}

	\item Otherwise, $Z_{r_i}^{(i)}(u) = \ctype$ for every token~$u$.
\end{itemize}
On subsequent time steps~$t > r_i$, $Z^{(i)}$ follows the rules of the annihilation dynamics (defined in \Cref{sec:cancellation}), using the same schedule as our protocol (meaning that the same pair of nodes interacts in both coupled processes at every time step).
We define the following good events regarding the behavior of each~$Z^{(i)}$. For every even~$i \geq 0$, define
\begin{equation*}
	\mathcal{E}_i := \left\{ \tclear\pa{Z^{(i)},\frac{1}{10}} \leq R  \right\},
\end{equation*}
and for every odd~$i \geq 1$, define
\begin{equation*}
	\mathcal{E}_i := \left\{ \textinction\left(Z^{(i)}\right) \leq R  \right\},
\end{equation*}
where the extinction time~$\textinction$ and the clearing time~$\tclear$ of the annihilation dynamics are as in \Cref{sec:clearing,sec:extinction}, respectively.
In addition, we define a good event for the global phase clock
\begin{equation*}
	\mathcal{F} := \left\{ \text{The  properties (i)--(iv) in \Cref{def:global_phase_clock} are satisfied for every~$i \le \log n$} \right\}.
\end{equation*}
The reason why we do not simply condition on~$\mathcal{F}$ and then reduce to the results of \Cref{sec:cancellation}, is that conditioning on~$\mathcal{F}$ would introduce dependencies that our analysis of the annihilation dynamics does not account for.
Nonetheless, by virtue of the coupling, we can show that although these events are not independent, they all happen simultaneously with a large enough probability.
\begin{lemma} \label{lemma:good_events}
	We have that
	\begin{equation*}
		\Pr \left[ \mathcal{F} \cap \bigcap_{i=0}^{2 \log n} \mathcal{E}_i \right] \geq 1 - \frac{1}{n^\kappa}.
	\end{equation*}
\end{lemma}
\begin{proof}
	We have that $\Pr[\mathcal{F}] \geq 1-1/n^{\kappa+1}$ by \Cref{def:global_phase_clock} and since~$(\phi_t)_{t \geq 0}$ is a $(R,\eta,\kappa+1)$-clock. For $i$ even, the same lower bound holds on~$\Pr[\mathcal{E}_i]$ by \Cref{lemma:clearing}, and by definition of~$R$ in \Cref{eq:R_def2}. Finally, for~$i$ odd, observe that~$Z^{(i)}$ is always initialized with bias at most~$1/10$ by definition. Therefore, the same lower bound holds by \Cref{thm:extinction-time} and the definition of~$R$ in \Cref{eq:R_def2}, and we can conclude the proof of \Cref{lemma:good_events} with the union bound.
\end{proof}

Eventually, we will show that if all these events happen, then \algone removes the minority opinion fast on the corresponding schedule. For that purpose, we first show that the bias must double every two phases.
Here the bias $\gamma_t$ for the protocol is defined as the number of tokens with type $\atype$ or $\smallatype$ minus the number of tokens with type $\btype$ or $\smallbtype$ dived by $n$, that is
\[
\gamma_t =\frac{|\omega_t^{-1}(\{\atype,\smallatype\})|-|\omega_t^{-1}(\{\btype,\smallbtype\})|}{n}.
\]
Note that $\gamma_0$ is the same as bias of the input.
Intuitively, the bias remains constant during cancellation phases, and during doubling phases, the bias doubles or the minority opinion tokens disappear completely.
However, as the phase transitions do not happen instantaneously, the two phases are actually overlapping. That is, typically during an interval $[r_i+R,r_{i+1})$ there will be a mix of tokens in phase $i$ and $i+1$. While two interacting tokens never apply rules of different phases in the same interaction and the tokens will always be in adjacent phases (assuming the phase clock is correct), we still need to handle this technicality that tokens may be in different -- but adjacent -- phases.

\begin{lemma} \label{lemma:main_induction}
	For every even~$i \le 2 \log n$, if
	\begin{equation*}
		\mathcal{F} \cap \bigcap_{j=0}^{i-1} \mathcal{E}_j
	\end{equation*}
	happens, then in time~$r_i$ either there are no more minority opinion tokens or the bias satisfies
	\begin{equation} \label{eq:doubling_gap}
		\gamma_{r_i} = 2^{i/2} \cdot \gamma_0.
	\end{equation}
\end{lemma}
\begin{proof}

We prove the claim by induction on $i$. Note that \Cref{eq:doubling_gap} holds trivially for~$i = 0$, since~$r_0 = 0$.
Now fix an \emph{even} integer~$i \in \{0,\ldots,2 \log n\}$ and assume that the statement holds for~$i$. We want prove the statement for $i+2$. To this end, assume that
\begin{equation*}
	\mathcal{F} \cap \bigcap_{j=0}^{i+1} \mathcal{E}_j \text{ happens.}
\end{equation*}
By~$\mathcal{F}$ and the {\em synchronization} property (c) of \Cref{def:global_phase_clock}, we have that for every~$t \in [r_i, r_i+R]$, $\phi_t(v) = i \bmod \Phi$.
First, note that there are no weak opinion tokens at time $r_i$. This is because either
\begin{itemize}[noitemsep]
\item $i=0$ and all opinion tokens are strong, or
\item $i\geq 2$ and then every weak opinion token transformed itself into a strong opinion token (of the same type) when transitioning from phase $i-1$ to phase $i$ during the interval $[r_{i-1}+R,r_i-1]$.
\end{itemize}
  By definition of \algone and since~$i$ is even, this implies that every token applies either the ``Initialization'' or the ``Cancellation'' rule in this time interval $[r_i, r_i+R]$, because  $i$ and $i \bmod \Phi$ have the same parity, since~$\Phi=4$ is even. In both cases, by construction of~$Z^{(i)}$ in \Cref{eq:coupling_even}, we have that for every~$t \in [r_i, r_i+R]$ the coupling satisfies
\begin{align*}
	\omega_t(u) = \atype &\iff Z_t^{(i)} = \atype, \text{ and} \\
	\omega_t(u) = \btype &\iff Z_t^{(i)} = \btype.
\end{align*}
By event~$\mathcal{E}_i$, this implies that at time~$r_i + R$ one of the two happens:
\begin{itemize}[noitemsep]
  \item either there are no tokens of type~$\btype$, or
  \item the number of strong opinion tokens is at most~$n/10$.
\end{itemize}
For the first case, suppose there is no token of type $\btype$ at time $r_i+R$.
Note that since all phases are even at that time, no token can be of type~$\smallbtype$.
Therefore, by definition of \algone, there is no token with type in~$\{\btype,\smallbtype\}$ in any subsequent time step~$t \geq r_i+R$ and the claim follows.

For the rest of the proof, we consider the second case. Suppose that at time $r_i+R$, the number of strong opinion tokens is at most~$n/10$.
By the event~$\mathcal{F}$ and the  synchronization property (c) of the phase clock, we have $\phi_t(v) = (i+1) \bmod \Phi$ for each $t \in [r_{i+1}, r_{i+1}+R]$. By definition of \algone and since~$(i+1)$ is odd, this implies that every token applies the ``Doubling'' rule in this interval.

Observe that there are no strong opinions at time $r_{i+1}+R$. To see why, by assumption there are at most~$n/10$ strong opinion tokens at time~$r_{i+1}$ and there are at most~$n/2$ clock tokens by \Cref{lemma:not_too_many_clocks}. Thus, the number of $\ctype$-tokens at time~$r_{i+1}$ is at least~$4n/10$. Therefore, the difference between the number $\ctype$-tokens and strong opinion tokens  is at least~$3n/10$, and $Z^{(i+1)}$ is as in \Cref{eq:coupling_odd}. Hence, we have for every~$t \in [r_{i+1}, r_{i+1}+R]$ that the coupling satisfies
\begin{align*}
	\omega_t(u) \in \{\atype,\btype\} &\iff Z_t^{(i+1)} = \atype, \text{and} \\
	\omega_t(u) = \ctype &\iff Z_t^{(i+1)} = \btype.
\end{align*}
By $\mathcal{E}_{i+1}$, this implies that at time~$r_{i+1}+R$, there are no strong opinion tokens.

Consider a strong opinion token $x$ at time $r_i$.
The token  will either be cancelled during the interval $[r_i,r_{i+1})$ or survive until time $r_{i+1}$.
If it gets canceled before $r_{i+1}$, a strong opinion of the opposite type will be cancelled in the same interaction, preserving the bias.
If the token survives until time $r_{i+1}$, it will double into two weak tokens of the same type by time $r_{i+1}+R$.%

Let $\chi$ be the number of $\atype$-tokens that got cancelled in the interval $[r_i, r_{i+1})$.
Note that $\chi$ is also the number of $\btype$-tokens that got cancelled in the same interval.
Now we have the equalities
\begin{align*}
	|\omega_{r_{i+1}+R}^{-1}(\{\atype,\smallatype\})| = |\omega_{r_{i+1}+R}^{-1}(\{\smallatype\})| &= 2\cdot (|\omega_{r_i}^{-1}(\{\atype\})|-\chi)=2\cdot (|\omega_{r_{i}}^{-1}(\{\atype,\smallatype\})|-\chi)\text{ and}\\
	|\omega_{r_{i+1}+R}^{-1}(\{\btype,\smallbtype\})| = |\omega_{r_{i+1}+R}^{-1}(\{\smallbtype\})| &= 2\cdot (|\omega_{r_i}^{-1}(\{\btype\})|-\chi)=2\cdot (|\omega_{r_{i}}^{-1}(\{\btype,\smallbtype\})|-\chi),
\end{align*} 
which give that $\gamma_{r_{i+1}+R}=2\cdot \gamma_{r_i}$.
Then, during the interval $[r_{i+1}+R,r_{i+2})$, opinion tokens in phase $i+1$ are only of type $\smallatype,\smallbtype$ or $\ctype$. This means that no doubling can occur; cancellations may occur from tokens already in phase $i+2$, but these interactions cannot change the bias, so $\gamma_{r_{i+2}}=\gamma_{r_{i+1}+R}$. Therefore, putting everything together, the induction hypothesis yields
\[
\gamma_{r_{i+2}} =\gamma_{r_{i+1}+R} = 2 \cdot \gamma_{r_i} = 2^{(i+2)/2} \cdot \gamma_{0}. \qedhere
\]
\end{proof}

Finally, we turn the exponential growth of the bias into the following bound. %

\begin{lemma} \label{thm:alg_one_correctness}
  With any input with bias $\gamma>0$, there are no more minority opinion tokens after
	\[
	O\left(\frac{\Delta}{\delta} \cdot \trel \log n \cdot\log \left( \frac{1}{\gamma} \right) \right)
	\]
	time steps with probability at least~$1-1/n^{\kappa}$.
\end{lemma}
\begin{proof}
  The event $\mathcal{F}\ \cap\ \bigcap_{i=0}^{2 \log n} \mathcal{E}_i$ happens with probability at least $1-1/n^{\kappa}$ by \Cref{lemma:good_events}. Conditioning on this event, \Cref{lemma:main_induction} yields that the bias doubles every 2 phases  or all $\smallbtype$ and $\btype$-tokens are eliminated. Since the bias is always at most~$1$, it cannot double more than~$2\log (1/\gamma)$ times, and therefore all $\smallbtype$ and $\btype$-tokens are eliminated after at most~$2\log (1/\gamma)$ phases. By $\mathcal{F}$ and the {\em bounded delay} property (b) of the phase clock, this means that there is no $\smallbtype$ or $\btype$-tokens after at most~$2 \eta R \log (1/\gamma)$ time steps, establishing the claim.
\end{proof}

\subsection{An always correct fast protocol: Algorithm 2} \label{sec:failure_detection}

The above algorithm can fail (with low probability) and does not yet ensure that all outputs agree.
In order to obtain almost sure correctness as well as finite expected stabilization time, we use the typical approach of using an always correct constant-state protocol as backup~\cite{elsaesser2018recent,alistarh2018recent,doty2022time,alistarh2018space}. For the sake of completeness, we detail the backup construction below.

We define \algtwo, which consists in running \algone in parallel with the 4-state protocol defined in \Cref{eq:4-state-majority}. We show that any critical failure can be detected locally by at least one node/token. If this happens, a broadcast is started which prompts every other node to give up the execution of \algone, and to adopt the output of the 4-state protocol instead.

\paragraph{Details of the backup construction.}
The tokens maintain three binary flags, called \abort, \awins and \bwins, each initialized to~$0$. When two tokens~$u$ and~$v$ interact and~$v$ has any flag set to~$1$, $u$ also sets the corresponding flag to~$1$.
In addition to \algone and the 4-state protocol, tokens apply the following \emph{flag update rules}:
\begin{itemize}
	\item[(i)] An $\atype$-token (resp.~$\btype$-token) that enters an even phase
	raises the \awins flag (resp. \bwins).

	\item[(ii)] An opinion token with state in $\{ \atype,\smallatype \}$ (resp. $\{ \btype,\smallbtype \}$) that encounters the \bwins flag (resp. \awins) raises the \abort flag.
	
	\item[(iii)] Two interacting tokens whose phase or iteration counter differ by more than 1, and that do not have the \awins or \bwins flag already raised, raise the \abort flag.
	
	\item[(iv)] If the interaction counter of a token reaches $2 \log n$, and this token does not have the \awins or \bwins flag already raised, it raises the \abort flag.
\end{itemize}

\paragraph{Outputs.}
A token with the \abort flag raised always adopts the output of the 4-state protocol running in the background. Otherwise, if a token has any of the \awins or \bwins flag raised, it outputs the corresponding opinion. In all other cases, the output of the 4-state protocol is used.
In particular, when none of the flags are set, the output of the 4-state protocol is used.

\paragraph{Analysis.}
We are now finally ready to prove our main upper bound.
First, we check that in the absence of failures, \algtwo has the same stabilization time as \algone.

\begin{lemma}\label{lemma:good-event-implies-stabilization}
  If the event $\mathcal{F}~\cap~\bigcap_{i=0}^{2 \log n} \mathcal{E}_i$ happens, then \algtwo stabilizes by time
	\[
	O\left(\frac{\Delta}{\delta} \cdot \trel \log n \cdot\log \left( \frac{1}{\gamma} \right) \right)
	\]
	with probability $1-1/n^\kappa$.
\end{lemma}
\begin{proof}
  Suppose the event $\mathcal{F}~\cap~\bigcap_{i=0}^{2 \log n} \mathcal{E}_i$ happens. By \Cref{thm:alg_one_correctness},
there are no more minority tokens by the stated time bound.
Note that conditioned on the above event, the flag update rule (iii) is never applied in the first~$2 \log n$ phases, since $\mathcal{F}$ happens. As in the proof of \Cref{lemma:main_induction}, no strong opinion token enters an even phase before all tokens with state in~$\{\btype, \smallbtype\}$ have been eliminated. When this happens, $\atype$-tokens keep duplicating during subsequent doubling phases until there are not enough $\ctype$-tokens. At this point, one $\atype$-token must enter an even phase and raise the \awins flag by flag update rule (i).

Since there are no minority opinion tokens left at this point, flag update rule (ii) may no longer be applied, and the \awins flag spreads without raising \abort flags. With probability at least $1-1/n^\kappa$, this broadcast spreads in $O(\trel \log n)$ steps by our choice of $R$ and \Cref{lemma:broadcast-faster-than-relaxing}. After that, flag update rule (iii) may no longer be applied and all tokens have the correct output. Finally, by~$\mathcal{F}$, all of this happens before any iteration counter reaches~$2 \log n$, so flag update rule (iv) is never applied.
\end{proof}

Next we show that with the backup construction, \algtwo is always correct. To this end, define the {\em potential} of a token~$u$ at time~$t$ as
\begin{equation*}
	\lambda_t(u) := \begin{cases}
		+2^{-\lfloor\vartheta_t(u)/2\rfloor} & \text{if } \omega(u) = \atype, \\
		-2^{-\lfloor\vartheta_t(u)/2\rfloor} & \text{if } \omega(u) = \btype, \\
		+2^{-\lfloor(\vartheta_t(u)+1)/2\rfloor} & \text{if } \omega(u) = \smallatype, \\
		-2^{-\lfloor(\vartheta_t(u)+1)/2\rfloor} & \text{if } \omega(u) =\smallbtype, \\
		0 & \text{otherwise.}
	\end{cases}
\end{equation*}
Let~$\lambda_t^\star := \sum_{u \in V} \lambda_t(u)$ be the total potential.
As we will see, this definition ensures that the total potential is constant in a correct execution.
Moreover, we can show that tokens locally detect if the total potential has been modified (at least until stabilization).
Similar arguments have already been used in previous backup constructions; see, e.g.,~\cite{elsaesser2018recent,alistarh2018space,doty2022time}.
\begin{lemma} \label{claim:potential_invariance}
	Either~$\lambda_t^\star = \lambda_0^\star$ for all $t \ge 0$ or a token raises the \abort flag at some point during the execution.%
\end{lemma}
\begin{proof}
Consider the interaction at a given time~$t+1$. We consider all possible interaction scenarios that can change the potential:
\begin{itemize}
	\item If an $\atype$-token and a $\btype$-token with the same (even) iteration counter cancel each other, then the total potential is unchanged.
	\item If an $\atype$-token~$u$ doubles into two $\smallatype$-tokens~$v$ and~$w$, necessarily~$\vartheta_t(u)$ must be odd, so
\[
\lfloor(\vartheta_t(u)+1)/2\rfloor = \lfloor\vartheta_{t+1}(v)/2\rfloor + 1 = \lfloor\vartheta_{t+1}(w)/2\rfloor + 1
\]
and the total potential is unchanged. The reasoning is symmetric for a $\btype$-token.
	\item If an $\smallatype$-token~$u$ becomes an $\atype$-token, then $\vartheta_t(u)$ is odd and~$\vartheta_{t+1}(u) = \vartheta_t(u)+1$. Therefore $\lfloor(\vartheta_t(u)+1)/2\rfloor = \lfloor\vartheta_{t+1}(u)/2\rfloor$ and the total potential is unchanged. Again, the case is symmetric for a $\smallbtype$-token becoming a $\btype$-token. %
	\item If an $\atype$-token and a $\btype$-token with {\em different} potential cancel each other, then their iteration counters must both be even, and must differ by at least~$2$. In that case, the total potential is changed but the tokens raise the \abort flag by flag update rule (iii).
	\item If an $\atype$-token enters an even phase, then the total potential is changed, but this token raises the \awins flag. If there is no token with state in~$\{\btype,\smallbtype\}$ at time~$t$, then the statement holds. Otherwise, this will result in the \abort flag being raised by rule (ii) later in the execution. Again, the reasoning for $\btype$-tokens is symmetric.
\end{itemize}
In all other cases, the iteration counters and the types of the tokens are not modified, and therefore, the total potential is unchanged; which concludes the proof of \Cref{claim:potential_invariance}.
\end{proof}

We are finally ready to prove \Cref{thm:fast_upper_bound}.
\begin{proof}[Proof of \Cref{thm:fast_upper_bound}]
The space complexity bounds follows from the fact that opinion tokens use $O(\log n)$ states and clock tokens use  $O\left( \log n \cdot \left( \log \frac{\Delta}{\delta} + \log \frac{\trel}{n} \right) \right)$ states by \Cref{thm:clock-thm}.

Next, we check that the protocol always stabilizes to a correct output.
Note that~$\lambda_0^\star>0$, since all tokens start with their iteration counter $\phi(\cdot)$ set to 0 and there are more $\atype$-tokens than $\btype$-tokens without loss of generality.
Therefore, \Cref{claim:potential_invariance} rules out the possibility that all majority opinion tokens
are eliminated without the \abort flag being raised, since that would imply~$0\geq \lambda_t^\star \neq \lambda_0^\star$.
In turn, this rules out the possibility that the \bwins flag is raised without the \abort flag being later raised, because one remaining $\atype$-token would eventually catch it by flag update rule (ii).
This leaves us with only three possibilities in any execution:
\begin{enumerate}
\item If \algone removes all minority opinion tokens, then no \abort flag is set and the outputs are correct. Hence, in this case, \algtwo stabilizes by \Cref{lemma:good-event-implies-stabilization}.
\item If \algone fails, then an \abort flag is propagated in the system  and all tokens switch back to the output of the 4-state protocol.
\item If \algone is indefinitely delayed (e.g., if the global phase clock breaks with low probability), then all tokens use the output of the 4-state protocol by default.
\end{enumerate}
Thus, the outputs will always stabilize to a correct value, as the 4-state protocol is always correct.

To get the tail bound on the stabilization time of \algtwo, note that  the probability of (2) or (3) happening is at most $1/n^\kappa$ by \Cref{lemma:good_events}.
Hence, with high probability, \algtwo stabilizes when \algone stabilizes by \Cref{lemma:good-event-implies-stabilization}.
To get the bound on the expected stabilization time of \algtwo, let
 $T_1$ be the stabilization time of \algone and $T_2$ be the stabilization time of the 4-state protocol. Since $\E[T_2] \in \poly(n)$, choosing sufficiently large $\kappa>1$ and writing $\mathcal{E} = \mathcal{F}~\cap~\bigcap_{i=0}^{2 \log n} \mathcal{E}_i
$ gives us that expected stabilization time of \algtwo is
\[
 \E\left[T_1 \mid \mathcal{E} \right] \cdot \Pr[\mathcal{E}] + \E[T_2 \mid \overline{\mathcal{E}}] \cdot (1- \Pr[\mathcal{E}]) \le \E[T_1 \mid \mathcal{E}] + 1.
\]
The running time bound then follows from \Cref{thm:alg_one_correctness}.
\end{proof}

\section{Conclusions}\label{sec:conclusions}

In this paper, we gave new bounds for exact majority on general graphs,  significantly improving time and space complexity compared to the state-of-the-art in non-complete graphs~\cite{alistarh2021fast,berenbrink2016plurality}. We believe that our contributions extend further, beyond the specific problem of exact majority consensus.

\paragraph{New tools for population protocols on graphs.}
We refined an existing analysis of the annihilation dynamics~\cite{draief2012convergence} to express the convergence time of the dynamics in terms of the relaxation time of the random walk on $G$. This can be used to derive useful bounds on the running time of other fundamental processes: here, we used coupling arguments to bound the length of the cancellation and doubling phases of our fast protocol, and the stabilization time of the 4-state protocol of B\'en\'ezit et al.~\cite{benezit2009interval}. We suspect that our approach can be extended to obtain improved bounds for, e.g., plurality consensus on graphs~\cite{berenbrink2016plurality}.

Finally, our new distributed phase clock may be reused to adapt other population protocols originally designed for the complete graphs to arbitrary graphs.
Indeed, our new phase clock can replace the power-of-two-choices clock used by Alistarh et al.~\cite{alistarh2021fast} in their general transformation of clique protocols to regular graphs, substantially improving the space complexity of their transformation.

\paragraph{Other asynchronous interaction models.}
Interestingly, our techniques do not fundamentally rely on the assumption that the scheduler is uniform over the edges. The arguments in \Cref{sec:cancellation} also extend
to cases where the distribution over the edges is not uniform, i.e., when the edges have different activation rates.
In such cases, the resulting bounds would simply depend on the relaxation time of the corresponding random walk induced by the scheduler.

Similarly, for the global phase clock, we have shown in \Cref{sec:clock_ticks_general} that our construction can be adapted to account for heterogeneous activation rates, and therefore, to non-uniform schedulers. To do so, we only need bounds on the minimum and maximum activation probabilities of the nodes, and to accept an overhead proportional to their ratio.

As a concrete example, consider the \emph{Poisson node clock} model~\cite{giakkoupis2016asynchrony}, where a node is selected uniformly at random and then contacts one of its neighbors at random.  For this model, the relaxation time is simply the relaxation time of the classic random walk slowed down by a factor of $\Theta(n)$. The activation rates are uniform, so our clock construction would have a negligible overhead.
Thus, our results can be translated to other interaction models, with fairly little effort.

\paragraph{Open problems.}
We conclude with two intriguing research questions prompted by our results:

\begin{enumerate}
\item What are the space-time trade-offs for majority consensus in general interaction graphs?
  In expanders, our protocol matches the known $\Omega(\log n)$ space lower bound for fast protocols in complete graphs~\cite{alistarh2018space}, as complete graphs are also (high-degree) expanders.
  However, is it possible to rule out near-time-optimal protocols that use \emph{sublogarithmic} space for some natural (sparse) graph classes excluding complete graphs?

\item In non-regular graphs, our complexity bounds depend on the degree imbalance $\Delta/\delta$,
 because we slow down all clocks, to ensure that nodes with higher-than-average degrees do not produce clock ticks too frequently. Is there a phase clock that could be used to eliminate this dependency, yielding bounds that depend only on $\trel$ also in highly degree-imbalanced~graphs?
\end{enumerate}

\subsection*{Acknowledgements}

This research was funded by the Deutsche Forschungsgemeinschaft (DFG, German Research Foundation) -- Project number 539576251.

\bibliographystyle{plainnat}
\bibliography{references}

\appendix

\section{Appendix: Omitted proofs}

\subsection{Proof of \Cref{lemma:broadcast-faster-than-relaxing}}\label{apx:broadcast-relax}

\broadcastrelax*
\begin{proof}
  Alistarh et al.~\cite{alistarh2025near} showed that there exists a function $C(\alpha)$ such that for any graph $G = (V,E)$ and any node $v \in V$ the broadcast time $\tbroadcast(v)$ from $v$ satisfies
  \begin{align*}
\Pr\left[\tbroadcast(v) \ge C(\alpha) \frac{m \log n}{\zeta} \right] \le \frac{1}{n^{\alpha}},
\end{align*}
  where $\zeta$ is the edge expansion of the graph.
  By \Cref{lemma:spectral-gap-sandwich}, we get that $\trel \ge m/\zeta$.
  Therefore,  we with $C(\alpha)$ we have
  \[
  \Pr[\tbroadcast(v) \ge  C(\alpha) \trel \log n ] \le \Pr\left[ \tbroadcast(v) \ge
  C(\alpha) \cdot \frac{m \log n}{\zeta}\right] \le \frac{1}{n^{\alpha}}. \qedhere
  \]
\end{proof}

\subsection{Proof of \Cref{lemma:spectral-gap-sandwich}}\label{apx:spectral-sandwich}

\spectralsandwich*
\begin{proof}
The Markov chain $X$ defined by $P$ is irreducible and periodic and its stationary distribution $\pi$ is uniform.
Define
\[
B(S, V \setminus S) = \sum_{x \in S, y \notin S} \pi(x) p_{xy}. %
\]
Recall that the bottleneck ratio~\cite[Eq.~(7.5)]{levin2017markov} of $X$ is given by
\begin{align*}
\Phi :&= \min\left\{ \frac{B(S, V \setminus S)}{\pi(S)} :  S \subseteq V, \pi(S) \le 1/2 \right\}.
\end{align*}
Since $\pi(x) = 1/n$ for any $x$ and $p_{x,y} = 1/(2m)$ for any $\{x,y\} \in E$, we get that 
\begin{align*}
\Phi &= \min \left\{ \frac{n |\partial S| }{2m \pi(S)} : S \subseteq V, \pi(S) \le 1/2 \right\} \\
     &=  \min\left\{ \frac{|\partial S| }{2m |S|} : S \subseteq V, |S| \le n/2 \right\} = \frac{\zeta}{2m}.
\end{align*}
The claim now follows as, $\trel$ is the inverse of the eigenvalue gap $1-\lambda_2$, which is bounded by the bottleneck ratio~\cite[Theorem 13.10]{levin2017markov} from both sides by
\[
\frac{\Phi^2}{2} \le 1 - \lambda_2 \le 2\Phi. \qedhere
\]
\end{proof}

\subsection{\texorpdfstring{Proof of \Cref{thm:individual_broadcast_proba_lb}}{Proof of Lemma}}\label{apx:broadcast-lb-lemma}

\tbrlemma*
\begin{proof}
First, we consider the case that~$D \leq \log n$, and re-use arguments from the proof of Lemma 11 in~\cite{alistarh2025near}. Let~$\{u_1,u_2\} \in E$ be in any edge of the graph. In any given step, the probability that the number of nodes influenced by either~$u_1$ or~$u_2$ increases from~$i$ to~$i+1$ is at most~$p_{i+1} = d \, i/m$. Hence, the number of steps before this happens stochastically dominates the geometric random variable~$Y_{i+1} \sim \mathrm{Geom}(p_{i+1})$. Therefore, and since initially only two nodes ($u_1$ and~$u_2$) are influenced by either~$u_1$ or~$u_2$,
\begin{equation} \label{eq:stochastic_domination}
		\tbroadcast(\{u_1,u_2\}) \succeq Y := \sum_{i=2}^{n-1} Y_{i+1}.
\end{equation}
We have
\begin{equation} \label{eq:lb_exp_1}
	\Exp[Y] = \sum_{i=2}^{n-1} \frac{m}{d \, i} = \frac{m}{\Delta} \pa{ H_{n-1} - 1 } \geq \frac{m}{d} \pa{\log (n-1) -1},
\end{equation}
where~$H_{n-1}$ is the~$(n-1)^{\text{th}}$ harmonic number. Since~$G$ is~$d$-regular, $m/d = n/2$ and for large enough $n$, the last inequality implies
\begin{equation} \label{eq:lb_exp_2}
	\Exp[Y] \geq \tfrac{1}{4} \cdot n \log n.
\end{equation}
We have
\begin{align*}
	\Pr \Big[ \tbroadcast(\{u_1,u_2\}) \leq \frac{\lambda}{4} \cdot n \log n \Big] &\leq \Pr \Big [ Y \leq \frac{\lambda}{4} \cdot n \log n \Big ] & \text{(by \Cref{eq:stochastic_domination})} \\
	&\leq \Pr \Big [ Y \leq \lambda \, \Exp(Y) \Big ] & \text{(by \Cref{eq:lb_exp_2})} \\
	&\leq \exp \pa{- \frac{2d}{m} \cdot \Exp[Y] \cdot c(\lambda) } & \text{(by \Cref{lemma:sum-of-geometric}(b))} \\
	&\leq \exp \pa{- 2 \, c(\lambda) \cdot (\log(n-1) - 1) }. & \text{(by \Cref{eq:lb_exp_1})}
\end{align*}
Since~$\lim_{\lambda \rightarrow 0} c(\lambda) = +\infty$, there exists a value~$\lambda_0 > 0$ such that for every large enough~$n$,
\begin{equation*}
	 c(\lambda_0) \cdot (\log(n-1) - 1) \geq \log n,
\end{equation*}
which implies the statement in \Cref{thm:individual_broadcast_proba_lb} with~$c_0 := \lambda_0/4$.

Now, we consider the case that~$D > \log n$.
For each $u \in V$, we define %
\begin{equation*}
	T_{D-1}(u) := \min_{\substack{ v\in V \\ \dist(u,v) = D-1 }} \{T(u,v)\}.
\end{equation*}
In the proof of Lemma 13 in~\cite{alistarh2025near}, it is shown (using the fact that~$D-1 \geq \log n$) that for every~$u \in V$,
\begin{equation} \label{eq:stolen_claim}
	\Pr\Big [T_{D-1}(u) \leq \frac{(D-1)m}{d e^3} \Big ] \leq \frac{1}{n^2}.
\end{equation}
By taking~$c_0 := 1/ (4 \, e^3)$, we have
\begin{equation*}
	\frac{(D-1)m}{d \, e^3} = \frac{(D-1)n}{2 e^3} \geq c_0 \,n \, D,
\end{equation*}
and \Cref{eq:stolen_claim} implies
\begin{equation*}
	\Pr[T_{D-1}(u) \leq c_0 \,n \, D ] \leq \frac{1}{n^2}.
\end{equation*}
Let~$u_1,v \in V$ such that $\dist(u_1,v) = D$, and let~$u_2 \in N(u_1)$. Note that~$\dist(u_2,v) \geq D-1$. By definition, we have
\begin{equation*}
	\tbroadcast(\{u_1,u_2\}) \geq \min \{T(u_1,v),T(u_2,v) \} \geq \min \{T_{D-1}(u_1),T_{D-1}(u_2) \},
\end{equation*}
and hence,
\begin{align*}
    \Pr[\tbroadcast(\{u_1,u_2\}) \leq c_0 \,n \, D ] &\leq \Pr[ \{T_{D-1}(u_1) \leq c_0 \,n \, D \} \cup \{ T_{D-1}(u_2) \leq c_0 \,n \, D \} ] \\
    &\leq \Pr[T_{D-1}(u_1) \leq c_0 \,n \, D ] + \Pr[T_{D-1}(u_2) \leq c_0 \,n \, D ] \leq \frac{2}{n^2},
\end{align*}
which concludes the proof of \Cref{thm:individual_broadcast_proba_lb}.
\end{proof}

\subsection{\texorpdfstring{Proof of \Cref{lemma:draief_vojnovic}}{Proof of Lemma XXX}} \label{apx:dv-proof}

\dvlemma*
\begin{proof}%
For~$u \in V$, let~$A_t(u) = \one{\{ X_t(u) = \atype \}}$ and~$B_t(u) = \one{\{ X_t(u) = \btype \}}$.
Note that
\begin{equation*}
	1-A_t(u)-B_t(u)=\one{\{ X_t(u) = \ctype \}},
\end{equation*}
and hence the Markov chain $(A_t,B_t)_{t \ge 0}$ over~$\{0,1\}^V \times \{0,1\}^V$ entirely characterizes the annihilation dynamics.
By definition and since~$Q$ is symmetric,
\begin{align*}
	\frac{d}{dt} \Exp[B_t(u)] = &-\sum_{v \neq u} q_{u,v} \Exp[B_t(u) \cdot A_t(v)] & (\underset{u}{\btype} + \underset{v}{\atype} \to \ctype + \ctype) \\
	&-\sum_{v \neq u} q_{u,v} \Exp[B_t(u) \cdot (1-A_t(v)-B_t(v))] & (\underset{u}{\btype} + \underset{v}{\ctype} \to \ctype + \btype) \\
	&+\sum_{v \neq u} q_{u,v} \Exp[ (1-A_t(u)-B_t(u)) \cdot B_t(v)], & (\underset{u}{\ctype} + \underset{v}{\btype} \to \btype + \ctype)
\end{align*}
where we have indicated which rule is responsible for each term.
Rearranging, this gives
\begin{align}
	\frac{d}{dt} \Exp[B_t(u)] &= - \pa{\sum_{v \neq u} q_{u,v}} \Exp[B_t(u)] + \sum_{v \neq u} q_{u,v} \Exp[ (1-A_t(u)) \cdot B_t(v) ] \nonumber \\
	&= q_{u,u} \Exp[B_t(u)] + \sum_{v \neq u} q_{u,v} \Exp[ (1-A_t(u)) \cdot B_t(v) ].  \label{eq:derivative_A}
\end{align}
Recall that~$\calA_t = X_t^{-1}(\atype)$ denotes the set of nodes in state~$\atype$ at time~$t$.
We can divide the execution into intervals~$[t_k,t_{k+1})$ during which~$\calA_t$ does not change; each~$t_k$ is a random variable and a stopping time for the Markov chain~$(A_t,B_t)_{t \ge 0}$. Let~$S_k := \calA_{t_k}$ be the set of nodes in state~$\atype$ during the~$k^{\text{th}}$ interval, and for a time~$t \geq t_k$, let~$\mathbf{E}_k$ be the expectation conditional on the event~$\{\calA_t = S_k\}$. For every~$t \in [t_k,t_{k+1})$, as when $u\in S_k$, $B_t(u)=1-A_t(u)=0$, we can rewrite \Cref{eq:derivative_A} into
\begin{align}
	\frac{d}{dt} \mathbf{E}_k[B_t(u)]
	&= \begin{cases} \sum_{v \in V} q_{u,v} \, \mathbf{E}_k[ B_t(v) ] & \text{if } u \in V\setminus S_k, \\ 0 & \text{if } u \in S_k. \end{cases} \label{eq:derivative_A_cond}
\end{align}
Consider the matrix~$R_{S_k}$ defined in \Cref{eq:RS_def}, with~$Q'_{S_k}$ denoting the restriction of~$Q$ to~$S_k$.
Note that when~$u \in S_k$, $\mathbf{E}_k[ B_t(u) ] = 0$, and from \Cref{eq:derivative_A_cond} we can trivially write
\begin{equation*}
\frac{d}{dt} \mathbf{E}_k[B_t(u)] = 0 = \lambda(Q_{S_k}') \cdot \mathbf{E}_k[ B_t(u) ].
\end{equation*}
Therefore, \Cref{eq:derivative_A_cond} can be written in matricial form as
\begin{equation*}
	\frac{d}{dt} \mathbf{E}_k[B_t] = R_{S_k} \cdot \mathbf{E}_k[B_t].
\end{equation*}
Solving this equation gives for every~$t \in [t_k,t_{k+1})$,
\begin{equation} \label{eq:diff_eq_solution}
	\mathbf{E}_k[B_t] = e^{ R_{S_k} (t-t_k) } \cdot \mathbf{E}_k[B_{t_k}].
\end{equation}
For $t\in [t_k,t_{k+1})$, let~$R^\star(t)$ be a random matrix defined from the stopping times~$t_0,\ldots,t_k$ as
\begin{equation*}
R^\star(t) = (t-t_k) R_{S_k} + \sum_{\ell=0}^{k-1} (t_{\ell+1}-t_\ell) R_{S_\ell}.
\end{equation*}
With this notation, \Cref{eq:diff_eq_solution} and the strong Markov property imply that for every~$t \in [t_k,t_{k+1})$,
\begin{equation} \label{eq:B_final_expression}
	\Exp[B_t] = \Exp \left[ e^{ R^\star(t) } \cdot \Exp[B_0] \right].
\end{equation}
We denote by~$\norm{\cdot}_2$ the~$\ell^2$-norm and by~$\norm{\cdot}$ the matrix  norm induced by~$\norm{\cdot}_2$; i.e., for~$y \in \mathbb{R}^V$ and~$M,N \in \mathbb{R}^{V\times V}$,
\begin{equation*}
	\norm{y}_2 := \pa{\sum_{u \in V} y_u^2}^{1/2} \text{and } \norm{M} := \sup_{x \in \mathbb{R}^V, \, \norm{x}_2 = 1 } \norm{M x}_2.
\end{equation*}
Moreover, we also have $\norm{My}_2\leq\norm{M}\norm{x}_2$ and $\norm{MN}\leq\norm{M}\norm{N}$.
When $M$ is symmetric, we can express its norm with its spectrum as $\norm{M}=\max\{|\lambda|\mid\lambda\in\spec(M)\}$, so we have $\norm{e^M}=\max\{e^\lambda\mid\lambda\in\spec(M)\}$.
Then \Cref{eq:B_final_expression} gives
\begin{align*}
	\norm{ \Exp[B_t] }_2 &\le \Exp\left[ \norm{ e^{ R^\star(t) } \cdot \Exp[B_0] }_2 \right]&\text{(Jensen's inequality)} \\
	&\le \Exp\left[ \norm{ e^{ R^\star(t) } } \cdot \norm{ \Exp[B_0] }_2 \right] \\
	&\le \Exp\left[ \norm{ e^{ (t-t_k) R_{S_k} } } \cdot \prod_{\ell=0}^{k-1} \norm{ e^{ (t_{\ell+1}-t_\ell) R_{S_\ell} } } \cdot \norm{ \Exp[B_0] }_2 \right].
\end{align*}
Recall that the bias~$(|\calA_t|-|\calB_t|)/n$ is constant and equal to~$\gamma$ for all~$t \geq 0$. In particular, $|\calA_t| \geq \gamma n$, and thus for every~$\ell \in \{0,\ldots,k-1\}$, $|S_\ell| \geq \gamma n$.
Therefore, by \Cref{lemma:bound_on_lambda_RS},
\begin{equation*}
	\norm{ e^{ (t_{\ell+1}-t_\ell) R_{S_\ell} } } = e^{(t_{\ell+1}-t_\ell) \lambda\big(R_{S_\ell}\big) } \le e^{-(t_{\ell+1}-t_\ell) \gamma/\trel }.
\end{equation*}
Therefore, the previous inequality implies
\begin{equation*}
	\norm{ \Exp[B_t] }_2 \le e^{- \gamma t / \trel} \cdot \norm{\Exp[B_0]}_2 \le \sqrt{n} \, e^{- \gamma t / \trel},
\end{equation*}
where the last inequality comes from the fact that~$B_t$ is a binary vector of size~$n$. Finally, we obtain that for every~$t \geq 0$,
\begin{equation*}
	\Pr\left[\textinctioncont > t  \right] = \Pr[B_t \neq \vec{0}] = \Pr \left[ \bigcup_{i \in V} B_t(i) = 1 \right] \leq \sum_{i \in V} \Exp[B_t(i)] \leq \norm{\Exp[B_t]}_2 \cdot \norm{\vec{1}}_2 \leq n  \, e^{- \gamma t / \trel},
\end{equation*}
where the second inequality is by Cauchy-Schwarz inequality.
Taking~$t = (\kappa+1) \trel \ln n / \gamma$ concludes the proof of \Cref{lemma:draief_vojnovic}.
\end{proof}

\end{document}